\documentclass[11pt, onecolumn, peerreviewca]{IEEEtran}

\usepackage{enumerate}
\usepackage{graphicx}
\usepackage{epsf}
\usepackage{subfigure}
\usepackage{amsmath}
\usepackage{amssymb}
\usepackage{array}
\usepackage{setspace}
\usepackage{makecell}
\usepackage{pifont}
\usepackage[utf8]{inputenc}
\usepackage[english]{babel}
\usepackage[linesnumbered,ruled]{algorithm2e}
\usepackage{algorithmic}
\usepackage{textcomp}
\usepackage{bbm}
\graphicspath{{Fig/}}

\usepackage{multirow}
\usepackage{amsthm}

\newtheorem{theorem}{Theorem}

\newtheorem{definition}{Definition}
\graphicspath{{Fig/},{fig/}}
\title{Hybrid Controlled User Association and Resource Management for Energy-Efficient Green RANs with Limited Fronthaul}

\author{Li-Hsiang Shen, Chia-Lin Tsai, Chia-Yu Wang, and Kai-Ten Feng \\
Department of Electrical and Computer Engineering \\
National Yang Ming Chiao Tung University, Hsinchu, Taiwan\\
gp3xu4vu6.cm04g@nctu.edu.tw, asfo.cm02g@nctu.edu.tw, wang91801115.cm03g@nctu.edu.tw,\\
and ktfeng@mail.nctu.edu.tw}

\begin{document}
\maketitle

\begin{abstract}
To alleviate green house effect, high network energy efficiency (EE) has increasingly become an important research target in wireless green communications. Therefore, the investigation for resource management to mitigate the co-tier interference in the small cell network (SCN) is provided. Moreover, with the merits of cloud radio access network (C-RAN), small cell base stations (SBSs) can be decomposed of a central small cell (CSC) and remote small cells (RSCs). To achieve the coordination, the split medium access control (MAC) based functional splitting is adopted with scheduler deployed at CSCs and retransmission functions left at RSCs. However, limited fronthaul has a compelling impact at RSCs due to requirements of user quality-of-service (QoS). Accordingly, a traffic control-based user association and resource allocation (TURA) scheme is proposed for a centralized resource management. To deal with the infeasibility to control all RSCs by CSC, we propose a hybrid controlled user and resource management (HARM) scheme. A CSC performs TURA for RSCs to mitigate intra-group interference within localized C-RANs, whereas the CSCs among separate C-RANs conduct cooperative resource competition (CRC) game for alleviating inter-group interference. Based on regret-based learning algorithm, the proposed schemes are analytically proved to reach the correlated equilibrium (CE). Simulation results have validated the effect of traffic control in TURA scheme and the convergence of CRC. Moreover, the comparison of the proposed TURA, HARM, and CRC schemes with the benchmark is revealed. It is observed that the TURA scheme outperforms the other schemes under ideal fronthaul control, whilst the proposed HARM scheme can sustain EE performance considering feasible implementation.   
\end{abstract}


\section{Introduction} 
\label{ch:Intro}
As the wireless traffic grows drastically \cite{1}, how to fulfill increasing user demands becomes a dominating issue for the next generation wireless communication systems. The deployment of small cell base stations (SBSs) has been promoted to resolve the above-mentioned concern over the past years owning to its advantages of low transmit power and low cost \cite{2}-\cite{3}. On the other hand, it is estimated that the energy consumption for information and communication technology (ICT) is rising at the rate of $15$-$20$ percentages per year. It is appraised that the ICT industry is responsible for $3\%$ of worldwide annual electrical energy consumption, which gives rise to $2$-$4\%$ of world's carbon dioxide emissions and severe impacts on the global environments \cite{4}\nocite{5}-\cite{6}. Therefore, GreenTouch \cite{7} suggests that not only improving the entire network's achievable capacity but also cutting back the power consumption of base stations is significant to network energy efficiency (EE) so as to reduce overall carbon footprint. Recently, the architecture of hyper-dense SBS deployment is viewed as a key solution to satisfy the huge amount of traffic demand while reducing the energy consumption. Since small cell networks (SCNs) shorten the distance between transmitters and receivers, the required energy for data transmission against pathloss can be reduced \cite{8}. Triggered by the merits of low transmit power and the tendency of dense deployment for high data rate requirements, the SCNs can be regarded as an energy-efficient wireless configuration for communication systems. 

Furthermore, the emerging virtualization and cloud technologies shift the network functions from radio resource control (RRC) and physical (PHY) layers at the edge of SBSs to a central processing unit, which is called the cloud radio access network (C-RAN). The functional split capability \cite{9} in C-RAN has grabbed the attention from both academia and industry \cite{10}-\cite{11}. The coordination among small cells in C-RAN can provide economic advantages and performance gain, including improved coordination, scalability enhancement, cost reduction, and more flexibility in network deployments. The technologies of virtual network functions (VNFs) run the aggregated small cell functions from different base stations in the virtual machines. Seven types of functional splits has been investigated \cite{12} between the central small cells (CSCs) and remote small cells (RSCs) with the quality requirements of fronthaul, which connects CSCs to RSCs. The applicable network functions and the restrictions of hardware in different scenarios of functional splits are also discussed. In \cite{13}, the upper layer network functions of small cells are virtualized in CSC in order to perform centralized management and coordination to serve its corresponding RCSs. On the other hand, the remaining functions reside in the edged small cells, namely RSCs, and autonomously perform lower layer functionalities. 
\begin{figure}
\centering
\includegraphics[width=1\columnwidth]{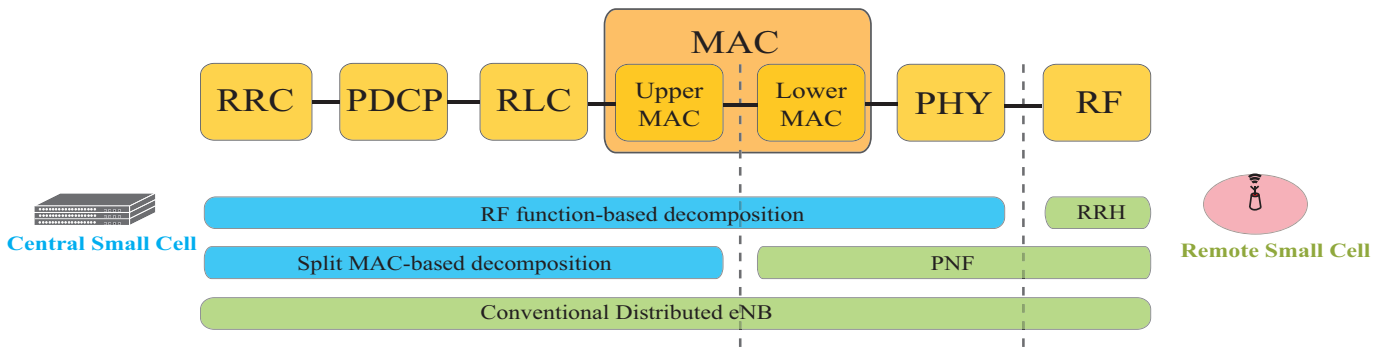}
\caption{Funtional split decompositions for C-RAN.} \label{Fig:functional}
\end{figure}

Fig. \ref{Fig:functional} shows the three types of functional splits extracted from \cite{13}-\cite{14} that will be discussed in this paper. First of all, conventional distributed small cells as illustrated in the bottom case of Fig. \ref{Fig:functional} conduct network functions autonomously, and the existing S1 solutions can support backhaul requirements with nationwide low cost IP networks. However, the limited capacity of non-ideal backhaul possesses a constraint to the SBSs and leads to performance bottleneck of achievable system throughput \cite{12}-\cite{13}. Under practical consideration of backhaul with limited capacity, there is a fundamental impact on user association while the backhaul links are overloaded. To guarantee the quality-of-service (QoS) of each user, the serving SBS with overloaded backhaul will offload some users to other SBS which is referred as traffic control \cite{15}. In \cite{16}, the authors aim at maximizing the weighted-sum rate for backhaul-constrained SCNs with carrier aggregation. Furthermore, it is crucial to appropriately allocate available frequency resources and SBS's transmit power to fulfill user's QoS requirement. Non-cooperative games in \cite{17}\nocite{18}-\cite{19} are adopted to enhance network capacity or EE by each SBS which selfishly determines the resource blocks (RBs) and power assignments based on its own utility functions, whereas the overall system performance will be degraded due to lacking of coordination. Another proposed framework in \cite{20} executes resource and power allocations by constructing a cooperative game between small cells for cross-tier and co-tier interference mitigation. Moreover, the scheme proposed in \cite{21} achieves optimal network EE based on a cooperative game for subcarrier assignment. Research in \cite{22} assigns subcarriers and allocates SBS's transmit power by evolutionary game theory, which analyzes the average interference between SBSs.

Another functional split virtualizes all the network functions in the CSC; while the RSCs only contain the radio frequency (RF) processing unit for data transmission and reception (as shown in the top-most case of Fig. \ref{Fig:functional}). This is the classic realization for C-RAN architecture and the VNFs are assigned to the most appropriate processor or hardware accelerator in CSCs to efficiently execute corresponding network functions of base stations. Although higher system performance can be achieved owing to full coordination among small cells, the requirements of low latency and ideal fronthaul will result in considerable expenditure to the network operators. Previous work in \cite{23} conducts RB and power allocations for SCN by a centrally controlled unit, and an optimization problem is formulated to solve for maximizing network EE. Nevertheless, the energy conservation and QoS requirements have not been taken into account in both literatures \cite{16} and \cite{24}. The data rate requirements and restrictions of backhaul capacity are simultaneously considered in the literatures \cite{14}, \cite{25} and \cite{26}. In \cite{14}, an energy-efficient resource allocation algorithm in multi-cell orthogonal frequency division multiple access (OFDMA) systems is proposed. The work in \cite{26} presents a joint resource allocation and admission control to minimize the sum of interference levels that macrocell can suffer from the small cells. Even though a near optimal solution can be obtained, the signaling overhead and computational loadings are potential drawbacks to the RF-based functional split management.

The above observations intuitively imply that there exists tradeoff between the overall system performance and the deployment cost of fronthaul links. According to the analysis in \cite{2}, medium access control (MAC) can be divided into upper MAC as the scheduler and lower MAC as the hybrid automatic repeat request (HARQ) mechanism. Additionally, functional split of MAC can deliver the benefits of centralization but only requires a small increase in transporting data in such network scenario. This is well-aligned with existing multi-vendor ecosystem for telecom operators based on the functional application platform interface. However, it is apparently infeasible for a CSC to control a huge amount of RSCs in realistic communication systems due to the hardware limitations. As a result, a practical scale for the implementation of dense SCNs is analyzed in this paper, where a CSC will be in charge of the radio resource management (RRM) for the RSCs in the localized small cell group such as a shopping mall or commercial building under the split MAC-based network functions. The interference management with limited fronthaul capacity will be invoked within a localized small cell group. Furthermore, the scheme for interference mitigation among the localized small cell groups will also be proposed in this paper. With the preponderance of split MAC-based functional splitting, this paper proposes a framework for subchannel and RSC transmit power allocation to maximize EE centrally by a CSC in its localized serving small cell group. The traffic control and small cell on/off mechanisms are also designed under the consideration of limited capacity for non-ideal fronthual. The traffic control occurs with overloaded fronthaul of serving RSC and this overloaded RSC will offload the users to one of the nearby RSCs, which remain enough fronthaul capacity to serve the user. Meanwhile, RSC will also tend to offload the associated users to others when its loading is low so as to turn off RSC for energy-saving. Hence, not only the required QoS under the circumstance of limited fronthaul capacity, but also the power conservation of RSCs can be achieved by the mechanism of user association. 

With the above statements, a joint optimization problem for traffic control-based user association and resource allocation (TURA) is proposed to be solved in this paper. To the best of existing knowledge, the merits of this joint optimization not only can reduce the computation loadings of CSC, but also reduce the power consumption to enhance the network performance. Unlike the ideal implementation for VNFs, which are co-located in a CSC, the interference alleviation for the group edge users among the localized small cell groups induces a challenge for the CSCs to coordinate their RRM since there is limited and asynchronous information exchanged between each other. Consequently, a cooperative resource competition (CRC) game for subchannel and transmit power allocation between CSCs is formulated in order to improve the network EE by effectively managing inter-group interferences. By regarding each CSC as a player in the game, a CSC can adopt its transmission alignments based on the observation of transmit power between each other and the received utility function. In the proposed CRC scheme, a distributed learning algorithm is employed by evaluating the regret value of different actions taken by the player. In summary, this paper proposes a hybrid controlled user association and resource management (HARM) scheme, which consists of both distributed and centralized RRM schemes for the C-RAN split MAC based functional splitting SCNs with the consideration of restricted fronthaul capacity.

The rest of this paper is organized as follows. The detailed descriptions about the system model and problem formulation of proposed HARM algorithm are provided in Section \textrm{II}. Section \textrm{III} illustrates the proposed TURA scheme for SCN in a localized C-RANs. Section \textrm{IV} formulates the CRC scheme among multiple localized C-RANs and adopted distributed algorithm to reach correlated equilibrium (CE). The performance of proposed framework is evaluated in Section \textrm{V}. Finally, Section \textrm{VI} draws the conclusions.

$\emph{Notations:}$ Denote bold capital letter as matrix and $|\cdot|$ as the absolute value. The operator $\text{max}(\cdot)$ returns the largest value in an array. $\mathbbm{1}(A)$ is an indicator function which is equal to 1 when the event $A$ happens and 0 otherwise. A random variable with value between $a$ and $b$ is generated by ${\rm rand}(a,b)$.
\section{System Model and Problem Formulation}
\label{sec:system model}
The detail descriptions for the architecture and operating process of proposed split MAC-based SCN are provided in the first and second subsections. Also, the energy-efficient optimization problem of resource allocation with traffic control and small cell on/off mechanisms under relevant constraints of subchannels, RSC transmit power, and capacity of fronthaul is formulated in the third subsection.
\subsection{System Model}
\label{System Model}
\begin{figure}
\centering
\includegraphics[width=1\columnwidth]{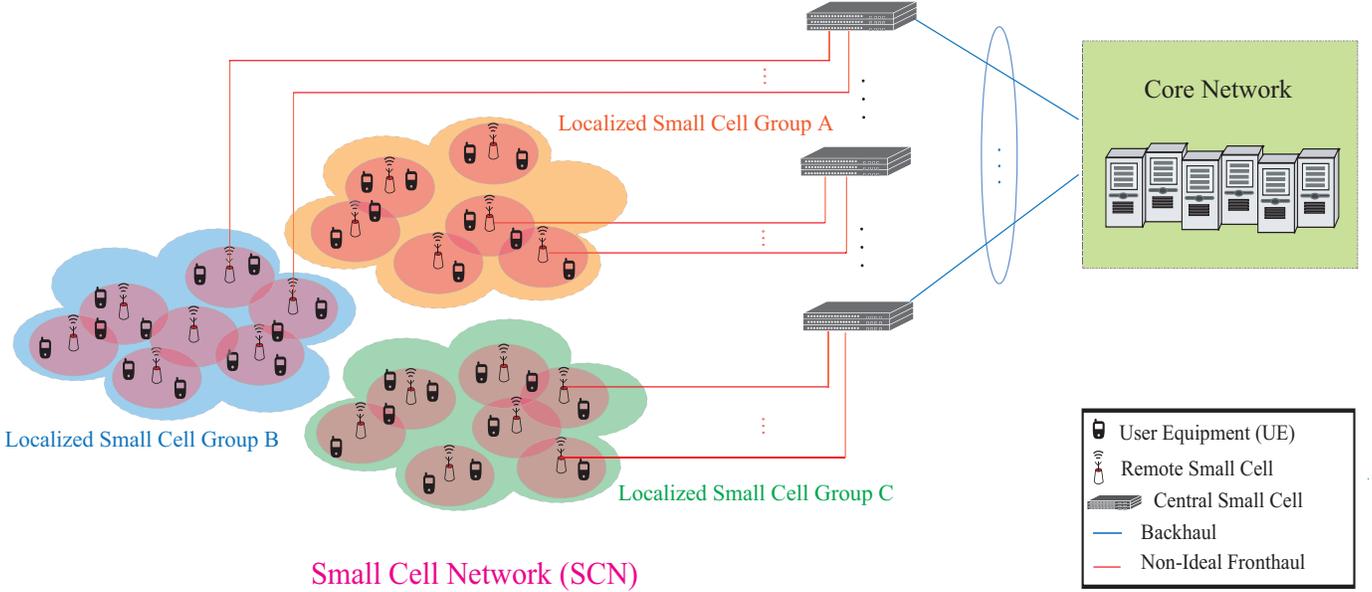}
\caption{C-RAN network achitecture for porposed split MAC-based SCN.} \label{Fig:SC}
\end{figure}
Fig. \ref{Fig:SC} illustrates a downlink orthogonal frequency-division multiplexing access (OFDMA) cellular network, which is divided into $\mathcal{C}=\{1,...,C\}$ localized small cell groups. There exits a set of $\mathcal{S}=\{1,...,S\}$ open access RSCs and a set of $\mathcal{K}=\{1,...,K\}$ serving UEs in the SCN. Each RSC connects to its serving CSC by non-ideal fronthaul links, which possesses the restriction of limited capacity in the rest of this paper within a localized small cell group. The CSCs in different small cell groups communicate with core network by backhaul links. In the proposed split MAC-based SCN, each CSC will conduct centralized resource allocation for its serving RSCs in a localized small cell group. In other words, a localized small cell group can be regarded as a localized C-RAN. It is considered that the channel state information at the transmitter (CSIT) can be measured by UEs and feedback to CSC through the RSCs.  Furthermore, there are $N$ available subchannels in the system and the bandwidth of each subchannel is $B$. All the RSCs share the entire frequency band, which leads to inter-cell interference between the RSCs and a subchannel can only be assigned to a single UE. As a result, the signal-to-interference-plus-noise ration (SINR) $\gamma_{c,s,k}^{n}$ of UE $k$ served by RSC $s$ in localized C-RAN $c$ on the subchannel $n$ is given as
\begin{align}
\label{SINR}
{\gamma}_{c,s,k}^{n}=\frac{p_{c,s,k}^{n}|g_{c,s,k}|^{2}}{I_{c,s,k}^{n}+N_{0}W},
\end{align}
where $g_{c,s,k}^{n}$ and $p_{c,s,k}^{n}$ are the deterministic channel gain and transmit power from RSC $s$ in the localized C-RAN $c$ to UE $k$ on subchannel $n$. $N_{0}$ is the power spectral density of additive white Gaussian noise (AWGN) and $W$ is the bandwidth of a subchannel. The term $I_{c,s,k}^{n}$ in the denominator of (\ref{SINR}) is represented in (\ref{I_{c,s,k}}), which include both the intra-group interference in a localized C-RAN (the first two terms) and the co-channel inter-group interference caused from other localized C-RANs to the localized C-RAN $c$ (the third term):
\begin{align}
\label{I_{c,s,k}}
I_{c,s,k}^{n}= {\sum\limits_{\substack{i=1\\i\neq{k}}}^{K}\phi_{c,s,j}\psi_{c,s,j}^{n}p_{c,s,j}^{n}g_{c,s,k}^{n}}+
\sum\limits_{\substack{i=1\\i\neq{s}}}^{S}{\sum\limits_{\substack{j=1\\j\neq{k}}}^{K}\phi_{c,i,j}\psi_{c,i,j}^{n}p_{c,i,j}^{n}g_{c,i,k}^{n}}
+\sum\limits_{\substack{t=1\\t\neq{c}}}^{C}\sum\limits_{\substack{i=1\\i\neq{s}}}^{S}{\sum\limits_{\substack{j=1\\j\neq{k}}}^{K}\phi_{t,i,j}\psi_{t,i,j}^{n}p_{t,i,j}^{n}g_{t,i,k}^{n}}
\end{align}
where $\phi_{t,i,j} \in \{0,1\}$ indicates whether user $j$ is associated with RSC $i$ in localized C-RAN $t$ and, on top of that, $\psi_{t,i,j}^{n} \in \{0,1\}$ decides the assignment of subchannel $n$ to user $j$ served by RSC $i$ in localized C-RAN $t$. Given the SINR calculated in (\ref{SINR}),  the achievable data rate $R_{c,s,k}^{n}$ for user $k$ served by RSC $s$ in the localized C-RAN $c$ on the subchannel $n$ can be formulated based on Shannon capacity as
\begin{align}
\label{Datarate_RB}
{R_{c,s,k}^{n}}=
{W}\log_{2}\left(1+{\gamma}_{c,s,k}^{n}\right).
\end{align}
Therefore, the sum-rate $R_{c}$ within a localized C-RAN $c$ is acquired as
\begin{align}
\label{Datarate_group}
{R_{c}}=
\sum\limits_{s=1}^{S}
\sum\limits_{k=1}^{K}
\phi_{c,s,k}
\sum\limits_{n=1}^{N}
\psi_{c,s,k}^{n}
R_{c,s,k}^{n}.
\end{align}
Furthermore, the transmit power of RSC $s$ in localized C-RAN $c$ is represented as
\begin{align}
\label{Txpower_cell}
P_{c,s}^{\rm (Tx)}
=\sum\limits_{k=1}^{K}\phi_{c,s,k}\sum\limits_{n=1}^{N}\psi_{c,s,k}^{n}p_{c,s,k}^{n}.
\end{align}
The users are initially connected to RSCs with the largest reference signal receiving power (RSRP) and may be offloaded to other cells based on the conditions of traffic load and available fronthaul capacity. A signal power overhead $P^{\rm (O)}$ is considered to reduce the ping-pong effect, which indicates handovers back-and-forth between two RSCs contributing to system over-loadings. Therefore, the power overhead in RSC $s$ resulted from traffic control is formulated as
\begin{align}
\label{Txpower_singaling}
P_{c,s}^{\rm (TC)}=P^{\rm (O)}\sum_{k=1}^{K}\phi_{c,s,k}\mathbbm{1}(|\phi_{c,s,k}-\bar{\phi}_{{c,s},k}|> 0),
\end{align}
where $\bar{\phi}_{{c,s},k}$ represents the initial state of user association. Based on (\ref{Txpower_cell}) and (\ref{Txpower_singaling}), the total power consumption $P_{c}$ within a localized C-RAN can be modeled as 
\begin{align}
P_{c}=&
\sum_{s=1}^{S}\mathbbm{1}\left(\sum_{k=1}^{K}\phi_{c,s,k}>0\right)
\left( P_{c,s}^{\rm (Tx)}+P_{c,s}^{\rm (TC)}+P^{\rm (CA)}\right) 
+\sum_{s=1}^{S}
\left(
1-\mathbbm{1}\left(\sum_{k=1}^{K}\phi_{c,s,k}>0\right)
\right) 
P^{\rm (CS)},
\label{Power-Consumption} 
\end{align}
where $P^{\rm (CA)}$ and $P^{\rm (CS)}$ are the circuit power consumptions of RSC in active mode and sleep mode respectively. Since switching RSC with low traffic load into sleep mode is an efficient approach to reduce the power consumptions of SBSs for green communications \cite{27}, $P^{\rm (CA)}$ and $P^{\rm (CS)}$ are consequently taken into considerations in the power model. The indicator function $\mathbbm{1}\left(\sum_{k=1}^{K}\phi_{s,k}>0\right)$ can be viewed as an RSC on/off strategy which is equal to $0$ when there is no user associated with RSC and RSC will fall into sleep mode for power-saving. The indicator function becomes $1$ when there is a single or multiple users associated with RSC and RSC is in the active mode. Accordingly, the EE of a localized C-RAN, which is defined as the ratio of total achievable data rate to the total power consumption of SBSs, can be expressed from (\ref{Datarate_group}) and (\ref{Power-Consumption}) as
\begin{align}
\label{EE_group}
{\eta_{c}}=
\frac{R_{c}}{P_{c}}.
\end{align}
\subsection{Operational Process for Proposed HARM Scheme}

\begin{figure}
\centering
\includegraphics[width=3.5in]{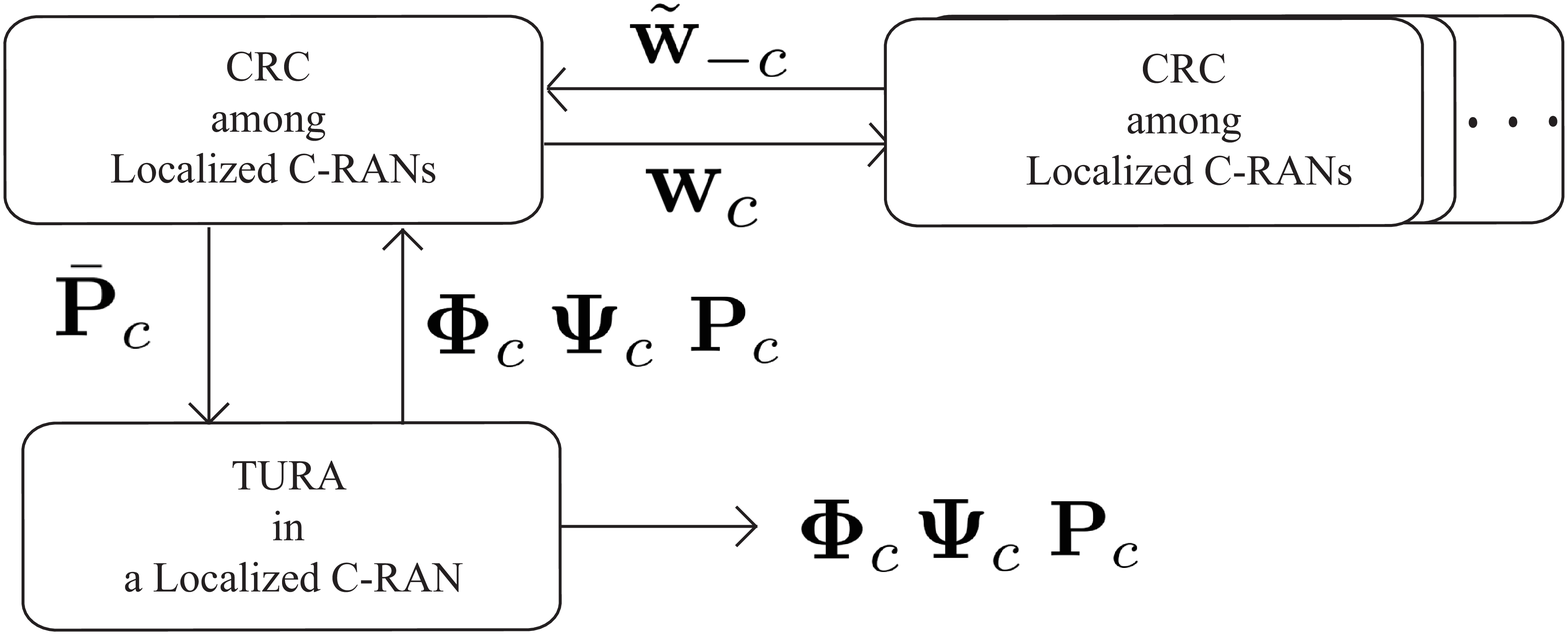}
\caption{Operating flow chart for interference mangngement of prposed HARM scheme.} \label{Fig:OP}
\end{figure}

The operating flow chart of proposed HARM system for mitigating both inter- and intra-group interferences is illustrated in Fig. \ref{Fig:OP}. Since the interference suffered from the UEs can be decomposed into both the intra-group interference within the same localized C-RAN and the inter-group interference from different localized C-RANs, the HARM scheme is proposed to overcome the above-mentioned problems by implementing the CRC and TURA algorithm repeatedly. By implementing the proposed TURA algorithm, the users will be associated with their serving RSC and allocated with proper configuration of subchannel and RSC transmit power considering the limited capacity of fronthaul and QoS requirement. The intra-group interference can be mitigated by the central control of CSC within each localized C-RAN. The decision strategies in TURA scheme including user association, subchannel allocation, and transmit power allocation within a localized C-RAN are respectively defined as $\mathbf{\Phi}_{c}=\{ \phi_{c,s,k}| 1\leq s \leq S, 1\leq k\leq K \}$, $\mathbf{\Psi}_{c}=\{\psi_{c,s,k}^{n}|1\leq s\leq S,1\leq k\leq K, 1\leq n\leq N\}$ and $\mathbf{P}_{c}=\{p_{c,s,k}^{n}|1\leq s\leq S,1\leq k\leq K, 1\leq n \leq N \}$. 

Furthermore, a CRC game is performed between CSCs based on the EE of their own localized C-RAN to alleviate the inter-group interference. The decision strategies $\mathbf{\Phi}_{c}$ , $\mathbf{\Psi}_{c}$, and $\mathbf{P}_{c}$ obtained in TURA scheme will be utilized by the localized CSC to determine the probability of strategy for resource allocation $\mathbf{w}_{c}$, and will be delivered from a localized C-RAN to other CRANs. It is considered that erroneous information $\tilde{\mathbf{w}}_{-c}$ will be received by a CSC from other localized C-RANs owning to asynchronous communications between CSCs. Noted that the detail descriptions of $\mathbf{w}_{c}$ and $\tilde{\mathbf{w}}_{-c}$ will be explained in Subsection IV.C.(3). Moreover, with the adoption of proposed CRC scheme, the set of upper bounds for transmit power on each subchannel in a localized C-RAN $c$ can be obtained as $\bar{\mathbf{P}}_{c}=\{\bar{P}_{c}^{n}|1\leq n \leq N\}$, where the total transmit power $\bar{P}_{c}^{n}$ on subchannel $n$ is $\bar{P}_{c}^{n}=\sum\limits_{s=1}^{S}\sum\limits_{k=1}^{K} \sum\limits_{\ell=1}^{L_{s}}\kappa_{c,s,k}^{n,\ell}$ and $\kappa_{c,s,k}^{n,\ell}$ is the constrained inter-group interference determined by proposed CRC scheme (will be shown in Subsection IV.A) and $L_{s}$ is the total power level. In other words, under the circumstance of bounded inter-group interference, each CSC will execute TURA scheme to determine user association, RSC transmit power allocation, and subchannel assignment for mitigation of intra-group interference. We can then effectively provide sum-rate enhancement and reduce power consumption under limited fronthaul capacity within a localized C-RAN. As a result, the proposed HARM scheme is performed by repeatedly executing TURA and CRC algorithms until it converges. 

\subsection{Problem Formulation}
The objective of this work is to maximize EE through subchannel assignment, power allocation, user association, and small cell on/off mechanisms under the constraints of data rate requirement of each UE, maximum transmit power allowance of each RSC, and limited fronthaul capacity. The optimization problem of resource allocation can be formulated to acquire the transmission policies for $\mathbf{\Phi}_{c}$, $\mathbf{\Psi}_{c}$, and $\mathbf{P}_{c}$ in order to improve network EE, which is stated as follows.

\begin{subequations}\label{Optimization}
\begin{align}
&\underset{\mathbf{\Phi},\mathbf{\Psi},\mathbf{P}}{\text{max}}
\eta(\mathbf{\Phi},\mathbf{\Psi},\mathbf{P})=\sum\limits_{c=1}^{C}\eta_{c}(\mathbf{\Phi}_{c},\mathbf{\Psi}_{c},\mathbf{P}_{c})
 \label{Objective Function} \\ 
    &\text{s.t} \quad P_{c,s}^{(\mathrm{Tx})}\leq P_{c,s}^{(\mathrm{max})}, &&\forall c,\forall s, \label{Constraint-1} \\
      &\quad\quad\sum\limits_{s=1}^{S}\phi_{c,s,k}\sum\limits_{n=1}^{N}\psi_{c,s,k}^{n}R_{c,s,k}^{n}\geq {R}_{c,k}^{(\mathrm{min})},&&\forall c,\forall s,\forall k, \label{Constraint-3} \\
  &\quad\quad\sum_{k=1}^{K}\phi_{c,s,k}\sum\limits_{n=1}^{N}\psi_{c,s,k}^{n}R_{c,s,k}^{n}\leq B_{c,s}^{(\mathrm{max})},&&\forall c,\forall s, \label{Constraint-4} \\
  &\quad\quad\sum\limits_{s=1}^{S}\phi_{c,s,k} \leq 1, &&\forall c,\forall k,\label{Constraint-6}\\ 
  &\quad\quad\phi_{c,s,k}, \psi_{c,s,k}^{n}\in\{0,1\},&&\forall c,\forall s,\forall k, \forall n,\label{Constraint-7} \\
   &\quad\quad p_{c,s,k}^{n}\geq 0,&&\forall c,\forall s,\forall k,\forall n.\label{Constraint-8}
\end{align}
\end{subequations}
The parameter $\mathbf{\Phi}$ in (\ref{Objective Function}) is the set of user association configuration and $\mathbf{\Phi}=\{\mathbf{\Phi}_{c}\in|1\leq c\leq C\}
$. The sets for decision policies for subchannel assignment and RSC power allocation are defined by $\mathbf{\Psi}=\{\mathbf{\Psi}_{c}\in|1\leq c\leq C\}
$ and $\mathbf{P}=\{\mathbf{P}_{c}\in|1\leq c\leq C\}$. $P_{c,s}^{(\text{max})}$ in (\ref{Constraint-1}) is the maximum transmit power of each RSC which restricts the sum of allocated power on all subchannels. (\ref{Constraint-3}) specifies that each user achieves its target data rate $R_{c,k}^{\text{(min)}}$ according to the QoS requirement. (\ref{Constraint-4}) depicts that the sum-rate of each RSC should be less than the allowance of fronthaul capacity $B_{c,s}^{(\text{max})}$. Furthermore, (\ref{Constraint-6}) describes that each user can be served only by a single RSC. The constraint in (\ref{Constraint-7}) indicates that both $\phi_{c,s,k}$ and $\psi_{c,s,k}^{n}$ are binary integer variables for user association and subchannel assignment, receptively. (\ref{Constraint-8}) defines the power allocation parameters to be non-negative values.
\section{Proposed Traffic Control-Based User Association and Resource Allocation (TURA) Scheme within a localized C-RAN}
\label{sec:TURA}
In this section, our proposed TURA scheme will be presented which centrally performs resource allocation by CSC for its corresponding RSCs in the localized C-RAN. The TURA scheme can mitigate the intra-group interference and provide higher network capacity by adopting efficient traffic offloading, user association, and resource allocation. Moreover, proper configuration of user association not only can satisfy the QoS requirement under limited fronthaul capacity but also can turn off the lightly loaded RSC to conserve the transmit power, which gives rise to achieve higher EE. Based on (\ref{Optimization}), the optimization problem for proposed TURA algorithm can be formulated for a localized C-RAN $c$ as follows.
\begin{subequations}\label{TCOptimization}
\begin{align}
    &\underset{\mathbf{\Phi}_{c},\mathbf{\Psi}_{c},\mathbf{P}_{c}}{\text{max}}
\tilde{\eta}_{c}(\mathbf{\Phi}_{c},\mathbf{\Psi}_{c},\mathbf{P}_{c})
\label{Objective Function for TC-UARA} \\ 
    &\text{s.t (\ref{Constraint-1})--(\ref{Constraint-8})}\\
    &\quad\sum\limits_{s=1}^{S}\sum\limits_{k=1}^{K}\phi_{c,s,k}\psi_{c,s,k}^{n}p_{c,s,k}^{n}\leq \bar{P}_{c}^{n}, &&\forall c,\forall n.
\end{align}
\end{subequations}
The objective function $\tilde{\eta}_{c}$ in (\ref{Objective Function for TC-UARA}) is EE for localized C-RAN $c$ based on estimated SINR, which is calculated as
\begin{align}
\label{eta_estimated}
\tilde{\eta}_{c}=\frac{\tilde{R}_{c}(\mathbf{\Phi}_{c},\mathbf{\Psi}_{c},\mathbf{P}_{c})}{P_{c}(\mathbf{\Phi}_{c},\mathbf{\Psi}_{c},\mathbf{P}_{c})},
\end{align}
where the estimated sum-rate $\tilde{R}_{c}$ for localized C-RAN $c$ is obtained as
\begin{align}
\label{sumrate_estimated}
\tilde{R}_{c}=\sum\limits_{s=1}^{S}\sum\limits_{k=1}^{K}\phi_{c,s,k}\sum\limits_{n=1}^{N}\psi_{c,s,k}^{n}W\log_{2}
(1+\frac{p_{c,s,k}^{n}g_{c,s,k}^{n}}{\tilde{I}_{c,s,k}+N_{0}W}).
\end{align}
The term $\tilde{I}_{c,s,k}$ in (\ref{sumrate_estimated}) is the estimated interference, which consists of intra-group interference and erroneous inter-group interference due to asynchronous information exchanged between localized C-RANs. $\tilde{I}_{c,s,k}$ can be acquired as
\begin{align}
\label{Inteference_estimated}
\tilde{I}_{c,s,k}^{n}= {\sum\limits_{\substack{i=1\\i\neq{k}}}^{K}\phi_{c,s,i}\psi_{c,s,i}^{n}p_{c,s,i}^{n}g_{c,s,k}^{n}}+
\sum\limits_{\substack{j=1\\j\neq{s}}}^{S}{\sum\limits_{\substack{i=1\\i\neq{k}}}^{K}\phi_{c,j,i}\psi_{c,j,i}^{n}p_{c,j,i}^{n}g_{c,j,k}^{n}}
+\sum\limits_{\substack{t=1\\t\neq{c}}}^{C}\sum\limits_{\substack{j=1\\j\neq{s}}}^{S}{\sum\limits_{\substack{i=1\\i\neq{k}}}^{K}\phi_{t,j,i}\psi_{t,j,i}^{n}\tilde{p}_{t,j,i}^{n}g_{t,j,k}^{n}}.
\end{align}
%
Since the proposed optimization problem of resource allocation in (\ref{TCOptimization}) contains nonlinear objective function, which is a ratio of three transmit decision policies $\mathbf{\Phi_{c}}$, $\mathbf{\Psi_{c}}$ and $\mathbf{P}_{c}$, it is difficult to solve this problem via conventional linear programming methods. On top of that, the optimization problem is non-convex with respect to $\mathbf{\Phi_{c}}$, $\mathbf{\Psi_{c}}$ and $\mathbf{P}_{c}$ since the discrete variables and co-channel interference are also taken into consideration. Generally speaking, the problem can be solved by exhaustive searching method, which tries all the allocated configuration for user association, subchannels, and, transmit power of RSC. However, the computational complexity grows exponentially with the number of RSCs. For conventional optimization methods to deal with this non-convex optimization problem, it is required to relax the discrete variables to continuous variables by adopting some programming techniques and transforming the original problem into a convex one. Nevertheless, the transformed optimization problem cannot be directly solved to obtain the near-optimal solutions \cite{28}-\cite{29}. 
  
To overcome the difficulty of dealing with the optimization problem in (\ref{TCOptimization}), the TURA algorithm with stochastic processes is proposed and described in this subsection. Particle swarm optimization (PSO), which motivated by the social behaviour of bird flocks or fish schooling has gained increasing popularity during the last decade due to its effectiveness in conducting difficult optimization tasks, especially in resource allocation of wireless systems. The potential solutions of resource allocation problem are called particles in PSO, and the particles will spread through the problem space to achieve the final resource configuration by considering historical data and the current best particle. Compared with the well-known genetic algorithm (GA), the main merits of PSO over GA relies on the momentum effect of velocity vectors for particle movement which can quickly move the current best solution of each candidate particle to the global best solution in order to result in faster algorithm convergence.
However, the solutions from PSO can easily be trapped in local optimums, which can be effectively alleviated by adopting the quantum-behaved particle swarm optimization (QPSO) \cite{30}-\cite{31}. A random number generator is utilized in QPSO with a certain probability distribution to simulate the particle trajectories in order to provide global convergence of particles. 
Furthermore, unlike PSO, the QPSO does not require velocity vectors for particles and also possesses fewer parameters to adjust, which makes it easier to be implemented in realistic wireless systems. The position of a particle in QPSO, i.e., the candidate solution $\{\mathbf{\Phi_{c}},\mathbf{\Psi_{c}},\mathbf{P_{c}}\}$ of (\ref{TCOptimization}), can be iteratively updated based on the particle fitness and evolution process for approaching the optimal solution.
  
  
\subsection{Fitness Function and Transformation for Unconstrained Form}  
In the proposed TURA scheme, each potential solution become a candidate by means of evaluating the quality of fitness function. In our considered optimization problem, the objective function in (\ref{TCOptimization}) is a key factor for the fitness function to decide how to allocate limited resource in wireless network. However, the fitness function is generally in an unconstrained form and a transformation from constrained objective function is required. To tackle this difficulty, the penalty function is adopted to transform the original optimization problem in (\ref{TCOptimization}) into an unconstrained one \cite{30}, where the fitness function is defined as the gain between reward and penalty functions. Note that the reward function is the objective function for achieving higher EE and the penalty function is the degree of the transmission policies that does not meet the constraints. Therefore, the fitness function can be formulated as
\begin{align}
{
{F}(\mathbf{\Phi_{c}},\mathbf{\Psi_{c}},\mathbf{P_{c}}) = {\tilde{\eta}_{c}}(\mathbf{\Phi_{c}},\mathbf{\Psi_{c}},\mathbf{P_{c}})-\alpha {\Delta}(\mathbf{\Phi_{c}},\mathbf{\Psi_{c}},\mathbf{P_{c}}),
\label{fitness} 
}
\end{align}	
where $\alpha$ is the penalty factor, which is a parameter for the particle to balance the fitness function between the EE performance and the penalty that does not satisfy constraints. 
Moreover, the ${\Delta}(\mathbf{\Phi_{c}},\mathbf{\Psi_{c}},\mathbf{P_{c}})$ in (\ref{fitness}) term represents the penalty function, which can be obtained as 
\begin{align}
{\Delta}&(\mathbf{\Phi_{c}},\mathbf{\Psi_{c}},\mathbf{P_{c}})=
\sum\limits_{s=1}^{S}\left[\max\left(0,P_{c,s}^{\rm (Tx)}- P_{c,s}^{\rm (max)}\right)\right]^2 
+\sum\limits_{s=1}^{S}\left[\max\left(0,\sum_{k=1}^{K}\phi_{c,s,k}\sum\limits_{n=1}^{N}\psi_{c,s,k}^{n}{R}_{c,s,k}^{n}- C_{s}^{\rm (max)}\right)\right]^2 \nonumber \\ &\quad\quad\quad\quad\quad\:\:+\sum\limits_{k=1}^{K}\left[\max\left(0,R_{c,k}^{\rm (min)}-\sum\limits_{s=1}^{S}\phi_{c,s,k}\sum\limits_{n=1}^{N}\psi_{c,s,k}^{n}{R}_{c,s,k}^{n}\right)\right]^2 
+\sum\limits_{k=1}^{K}\left[\max\left(0,\sum\limits_{s=1}^{S}\phi_{c,s,k} - 1\right)\right]^2 \nonumber\\
&\quad\quad\quad\quad\quad\:\:+\sum\limits_{s=1}^{S}\sum\limits_{n=1}^{N}\left[\max\left(0,\sum\limits_{k=1}^{K}\psi_{c,s,k}^{n} - 1\right)\right]^2 
+\sum\limits_{s=1}^{S}\sum\limits_{k=1}^{K}\sum\limits_{n=1}^{N}\left[\max\left(0,-p_{c,s,k}^{n}\right)\right]^2 \nonumber\\
&\quad\quad\quad\quad\quad\:\:+\sum\limits_{s=1}^{S}\sum\limits_{k=1}^{K}\left[\left(\phi_{c,s,k}\right)^{2}-\phi_{c,s,k}\right]^2 
+\sum\limits_{s=1}^{S}\sum\limits_{k=1}^{K}\sum\limits_{n=1}^{N}\left[\left(\psi_{c,s,k}^{n}\right)^{2}-\psi_{c,s,k}^{n}\right]^2.
\label{penalty} 
\end{align}
It can be observed from (\ref{penalty}) that the value of penalty function relies on the gaps between the transmission policies and the corresponding constraints. The more constraints unsatisfied, the higher penalty the particle needs to suffer from conducting this decision of transmission policies, which impacts the direction of path to globally best position or is screened out when comparing with its historical particles during the evolution process.
\subsection{Operation Process for TURA scheme}
Since the updating process for each variable is implementing by the same manner, let $\mathbf{X_{c}}$ generalizes the decision policies $\{\mathbf{\Phi_{c}},\mathbf{\Psi_{c}},\mathbf{P_{c}}\}$
for simplification. In each iteration $t$, there are $I$ candidate solutions of optimization problem in (\ref{TCOptimization}) to be chosen from the searching space, where the searching space is referred as all the possible solution sets. The detailed descriptions of TURA scheme are given as follows.  

\subsubsection{Initialization}
All the $I$ candidate solutions are initialized at the beginning. For notational simplicity, the $i$-th candidate solution at iteration $t$ is denoted as $\mathbf{X}_{i}(t)$. The fitness function in (\ref{fitness}) of each candidate solution will be calculated, and the threshold of fitness function $F^{\text{th}}$ for the algorithm is determined to achieve the convergent solution. 

\subsubsection{Evolution}
After the initialization of $I$ particles, we can obtain the best solution of candidate $i$-th particle in iteration $t$, $\mathbf{X}^{\rm (HB)}_{i}(t)$, which contributes to the largest value of fitness function in (\ref{fitness}) in its history. Moreover, the global best solution among all $I$ particles at iteration $t$, which is denoted as $\mathbf{X}^{\rm (GB)}(t)$ can be also acquired. The weighted mean of $I$ elite candidate solutions at iteration $t$ can be calculated as
\begin{align}
\mathbf{X}^{\rm (Mean)}(t) = \frac{1}{I}\sum_{i=1}^{I}{\mathbf{X}}^{\rm (HB)}_i(t).
\label{mean_position}
\end{align}
As a result, the $i$-th candidate solution in iteration $t$ is updated by the evolution equation \cite{32}-\cite{33} as follows.
\begin{equation}
\mathbf{X}_{i}(t+1) =
\left\{
\begin{split}
	&\mathbf{X}^{\rm (SP)}_{i}(t) + \beta(t) \left\vert\mathbf{X}^{\rm (Mean)}(t)-\mathbf{X}_{i}(t)\right\vert\cdot\ln\left(\frac{1}{\mu}_{i}(t)\right), \qquad \varepsilon_{i}(t)>0.5,\\
	&{\mathbf{X}}^{\rm (SP)}_{i}(t) - \beta(t) \left\vert\mathbf{X}^{\rm (Mean)}(t)-\mathbf{X}_{i}(t)\right\vert\cdot\ln\left(\frac{1}{\mu}_{i}(t)\right), \qquad \varepsilon_{i}(t)\leq0.5,\\
\end{split}
\right.
\label{evolution}
\end{equation}
where $\mu_{i}(t)={\rm rand}(0,1)$ and $\varepsilon_{i}(t)={\rm rand}(0,1)$. ${\mathbf{X}}^{\rm (SP)}_{i}(t)$ in (\ref{evolution}) is the attractor between local and global optimal solutions for candidate $i$ in iteration $t$, which is given by
\begin{align}
{\mathbf{X}}^{\rm (SP)}_{i}(t) &= \lambda_{i}\cdot (t)\mathbf{X}^{\rm (HB)}_{i}(t)+\left(1-\lambda_{i}(t)\right) \cdot \mathbf{X}^{\rm (GB)}(t), 
\label{SP}
\end{align}
where $\lambda_{i}(t)={\rm rand}(0,1)$.
On the basis of the evolution equation in (\ref{evolution}), the candidate solutions can be attained based on ${\mathbf{X}}^{\rm (SP)}_{i}(t)$, ${\mathbf{X}}^{\rm (Mean)}_{i}(t)$, and $\mathbf{X}_{i}(t)$. The new starting point $\mathbf{X}_{i}^{\rm (SP)}(t)$ for the candidate solution at next iteration $t+1$ depends on the best solution in the history of $i$-th candidate and the global best solution among all the $I$ candidates at iteration $t$. The parameter $\lambda_{i}(t)$ in (\ref{SP}) is a random variable for the candidate at next iteration to start from a new position, which cross-correlates the historically best solution with the global best solution instead of starting from either ${\mathbf{X}}^{\rm (HB)}_{i}(t)$ or ${\mathbf{X}}^{\rm (GB)}(t)$ unilaterally. Moreover, the second term $\beta(t) \vert\mathbf{X}^{\rm (Mean)}(t)-\mathbf{X}_{i}(t)\vert\cdot\ln(\frac{1}{\mu}_{i}(t))$ in the second term of (\ref{evolution}) affect the speed that the position of next $i$-th candidate solution converging to the final solution. The term of $\beta(t)$ is a coefficient that influences the convergence speed of the algorithm \cite{1} which is given by 
\begin{align}
\beta(t)=\left( \beta^{\rm (max)}-\beta^{\rm (min)}\right)  \cdot \frac{T-t}{T}+\beta^{\rm (min)},
\label{beta}
\end{align}
where $\beta^{\rm (max)}$ and $\beta^{\rm (min)}$ are respectively the maximum and minimum searching ranges in the solution space. The maximum number of iterations is denoted as $T$. It can be observed from (\ref{beta}) that $\beta(t)$ is linearly decreasing with $t$ since $\mathbf{X}_{i}(t)$ is approaching convergence for large $t$. The difference between $\mathbf{X}^{\rm (Mean)}(t)$ and $\mathbf{X}_{i}(t)$ also affects the speed of convergence. If the difference is large, which means this candidate solution is far away from current average position, this candidate solution at the next iteration should accelerate to converge, and vice versa. The term $\mu_{i}(t)$ is a random variable to alleviate the next candidate falling into local optimum. In addition to $\mu_{i}(t)$, $\varepsilon_{i}(t)$ is also set to reduce the probability that the next candidate solution trapped into local extremes, which can be regarded as the direction for the starting point of next candidate to search the potential solutions. 
\subsubsection{Convergence}
Ultimately, the proposed TURA algorithm completes once achieving the terminated condition. There are two terminated conditions where the first condition is the maximum number of iterations, i.e., $t=T$, whilst the second one is to reach the stop criterion. A gap ratio $\tau_{i}(t)$ is defined as
\begin{align}
\label{Stop_tau}
\tau_{i}(t) = \frac{\left\vert F(\mathbf{X}^{\rm (HB)}_{i}(t))-F(\mathbf{X}^{\rm (GB)}(t)) \right\vert}{F(\mathbf{X}^{\rm (GB)}(t))},
\end{align}
and the convergence is achieved if all the gap ratio values $\tau_{i}(t)$ are smaller than a convergence threshold $F^{\rm (th)}$ as
\begin{align}
\label{Stop}
\sum\limits_{i=1}^{I}\mathbbm{1}\left(\tau_{i}(t) \leq F^{\rm (th)}\right) = I.
\end{align}
The gap ratio $\tau_{i}(t)$ in $\eqref{Stop_tau}$ represents the distance between the $i$-th elite candidate solution and the global best solution in iteration $t$. $\eqref{Stop}$ denotes that the fitness values of all elite candidate solutions are closed enough to the global best solution, i.e., the fitness function in $\eqref{fitness}$ converges.

The overall procedure for proposed TURA scheme within a localized C-RAN is demonstrated in Algorithm \ref{alg:QPSO}. Moreover, the computational complexity of proposed TURA algorithm is $\mathcal{O}(I \times T \times S \times K \times N)$, which is linear to the numbers of candidate solutions, iterations, SCs, users, and subchannels \cite{31}. The performance gain provided by TURA will be illustrated via simulation in Section \ref{sec:PE}.

\begin{algorithm}[!tb]
\caption{Proposed TURA Algorithm}
\SetAlgoLined
\DontPrintSemicolon
\label{alg:QPSO}
\begin{algorithmic}[1]
\STATE Initialization: \\
	1) Set the number of candidate solutions $I$\\
	2) Set the maximum number of iterations $T$\\
	3) Set the iteration counter to $t=1$\\
	4) Set the threshold of convergence $F^{\rm (th)}$\\
	5) Initialize the candidate solution $\mathbf{X}_{i}(1)$ for all $i$\\ 
	6) Initialize $\mathbf{X}^{\rm (HB)}_{i}(1)=\mathbf{X}_{i}(1)$\\
	7) Initialize $\mathbf{X}^{\rm (GB)}(1)$ by selecting the best solution from $\mathbf{X}^{\rm (HB)}_{i}(1)$\\ 
	8) Initialize CONVERGENCE = \textbf{FALSE}\\
\REPEAT
\STATE Calculate the mean position of candidate solution $\mathbf{X}^{\rm (Mean)}(t)$ by $\eqref{mean_position}$ \\
\STATE Calculate the convergence speed $\beta(t)$ by $\eqref{beta}$\\
\FOR{$i=1,2,\dots,I$}
\STATE Calculate the local attractor $\mathbf{X}^{\rm (A)}_i(t)$ by $\eqref{SP}$\\
\STATE Update the candidate solution $\mathbf{X}_i(t+1)$ by $\eqref{mean_position}$--$\eqref{SP}$\\
\STATE Calculate the value of fitness function $F(\mathbf{X}_i(t+1))$ by $\eqref{fitness}$\\
\IF {${F}(\mathbf{X}_{i}(t+1))>{F}(\mathbf{X}_{i}(t))$} 
\STATE $\mathbf{X}^{\rm (HB)}_i(t+1)=\mathbf{X}_{i}(t+1)$ 
\ELSE 
\STATE $\mathbf{X}^{\rm (HB)}_i(t+1)=\mathbf{X}^{\rm (HB)}_i(t)$ 
\ENDIF\\
\STATE Calculate the value of fitness functions $F(\mathbf{X}^{\rm (HB)}_i(t+1))$ and $F(\mathbf{X}^{\rm (GB)}(t))$ by $\eqref{fitness}$\\
\IF {$F(\mathbf{X}^{\rm (HB)}_{i}(t+1))>F(\mathbf{X}^{\rm (GB)}(t))$} 
\STATE $\mathbf{X}^{\rm (GB)}(t+1)=\mathbf{X}^{\rm (HB)}_{i}(t+1)$ 
\ELSE 
\STATE $\mathbf{X}^{\rm (GB)}(t+1)=\mathbf{X}^{\rm (GB)}(t)$ 
\ENDIF
\STATE Calculate the value $\tau_{i}(t+1)$ by $\eqref{Stop_tau}$
\ENDFOR
\IF {$\sum\limits_{i=1}^{I}\mathbbm{1}\left(\tau_{i}(t+1) \leq F^{\rm (th)}\right) = I$}
\STATE CONVERGENCE = \textbf{TRUE}
\ENDIF
\STATE Increment of $t\leftarrow t+1$ 
\UNTIL{$t=T$ or CONVERGENCE = \textbf{TRUE}}
\end{algorithmic}
\end{algorithm}
\section{Proposed Cooperative Resource Competition (CRC) Game among Localized C-RANs}
\label{sec:CRC}
Consider the hardware limitation, the proposed TURA scheme cannot afford to control a large number of RSCs in a dense small cell network. Consequently, the entire network can be viewed as a gathering network of many localized C-RANs, where a CSC will be in charge of the RRM for RSCs within a group. The CRC scheme of subchannel assignment and transmit power allocation for RSCs among localized C-RANs is proposed and described in this section. After eliminating the intra-group interference by resource allocation within each localized C-RAN based on TURA, each CSC conducts CRC scheme to further alleviate the inter-group interference so as to achieve higher system performance. Since there is limited and asynchronous information exchanged between the localized C-RANs, it is difficult for all the CSCs to coordinate their RRM under centrally controlled operation. Hence, A distributed management is adopted to overcome above-mentioned difficulty based on game theory, where the resource competition between localized small cells can be formulated as a cooperative game. Each CSC only requires limited information about the probabilities of chosen actions from other CSCs. We design a learning algorithm based on EE of each localized C-RAN for the CSCs to determine their transmit actions, which consist of subchannels assignment and transmit power allocation of their serving RSCs.
\subsection{Cooperative Game-Based Resource Competition Game}
In this subsection, the formulation of proposed CRC scheme among localized small cell groups is analyzed and investigated. Based on cooperative game theory, the proposed CRC game in normal form can be denoted as follows.
\begin{align}
\label{Game theoretical model}
\mathcal{G}=\left(\mathcal{C},\{\mathcal{A}_{c}\}_{c\in \mathcal{C}},\{U_{c}\}_{c \in \mathcal{C} }\right),
\end{align}
where the CSC set $\mathcal{C}=\{1, ..., C\}$ denotes the set of players in a game, who compete for the subchannel set $\mathcal{N}=\{1,...,N\}$ and resource for transmit power of RSCs, i.e., the CSCs are the chiefs for action decision. Also, $\mathcal{A}_{c}$ is the action space of CSC $c$ for power allocation vectors. Let $\kappa_{c,s,k}^{n}=\phi_{c,s,k}^{n}\times p_{c,s,k}^{n}$ be the product indicator of subchannel assignment and allocated transmit power, which determines whether or how much power UE $k$ will be allocated on subchannel $n$ from RSC $s$ in localized C-RAN $c$. A finite and discrete action space is a significant requirement for cooperative games. As a result, to form a finite and discrete action space $\mathcal{A}_{c}$ for each player, let $L_{c}\in\mathbb{N}$ be the number of discrete power and $\kappa_{c,s,k}^{n,\ell}$ represents $\ell$-th transmit power level from RSC $s$ to UE $k$ over subchannel $n$, where $n\in \mathcal{N}, \ell \in \mathcal{L}_{c}$ and $\mathcal{L}_{c}=\{1, ..., L_{C}\}$. Notice that if $\phi_{c,s,k}^{n}=0$, zero transmit power will be allocated and $\kappa_{c,s,k}^{n,\ell}=0$. Thus, we define $\kappa_{c,s,k}^{0,0}$ as no subchannel is assigned to the user and the action space of RSC $s$ which is then expressed as
\begin{align}
\label{Action space of a CSC}
\mathcal{A}_{c}=\kappa_{c,s,k}^{0,0}\cup\left\lbrace \kappa_{c,s,k}^{n,\ell}:n\in\mathcal{N},\ell\in\mathcal{L}_{C}\right\rbrace.
\end{align}
Denote $\mathcal{A}$ by $\mathcal{A}=\mathcal{A}_{1}\times\cdots\times\mathcal{A}_{C}$ as the entire action space of all the players. Moreover, the last term $U_{c}$ of (\ref{Game theoretical model}) is the utility function of CSC $c$. In game theory, players will choose their actions based on their own utility functions so as to obtain the maximum reward. Since the major purpose of this paper is to maximum EE of each localized C-RAN and the performance of overall system can be further enhanced with the cooperation between CSCs, where each CSC will determine its strategy based on higher EE gain. Accordingly, $\eta_{c}$ can be regarded as the utility function for each CSC to decide its transmission strategy for resource allocation and $U_{c}$ can be described as
\begin{align}
\label{U_{c}}
U_{c}=\tilde{\eta}_{c}(\mathbf{a}_{c}),
\end{align}
where $\mathbf{a}_{c}=\{\kappa_{c,s,k}^{n,\ell}|1\leq s \leq S,1\leq k \leq K,1\leq c \leq C, 1\leq n\leq N,0\leq \ell \leq L_{c}, s\in\mathbb{Z}^{+}, k\in\mathbb{Z}^{+}, c\in\mathbb{Z}^{+}, n\in\mathbb{Z}^{+}, \ell\in\mathbb{Z}^{+}\}$ represents the vector of actions taken by CSC $c$.
\subsection{Existence of Correlated Equilibrium in Proposed Scheme} 
In the proposed CRC game $\mathcal{G}$, each CSC aims to maximize its utility function cooperatively and further improve the EE performance of entire system by choosing an optimal transmission action $\mathbf{a}_{c}$, which is the solution of (\ref{Optimization}). Most of the existing works investigate the potential solutions of players' strategies by the stable point, which can lead to no play obtaining higher utility gain by changing their determined actions on the Nash equilibrium (NE) \cite{34}. However, a higher utility gain can be acquired by the players cooperatively deciding their strategies via information sharing, which is called cooperative game. In the cooperative games, the stable point that no player will unilaterally deviate from the selected action to other ones can be held by the correlated equilibrium (CE) in the game theory. The concept of CE is defined as follows.
\begin{definition}[Correlated Equilibrium]
Denote $\Delta\mathcal{A}$ as the set of all probability distributions over the finite action space $\mathcal{A}$. The probability of correlated strategy $A$ is given by $P(A)$, where $(P(A))_{A\in\mathcal{A}}\in\Delta\mathcal{A}$. On the condition that $\forall \mathbf{a}_{c}\in\mathcal{A}_{c}$ and $\forall c \in \mathcal{C}$, the CSC $c$ will choose strategy $\mathbf{a}_{c}$ rather than any other strategy $\tilde{\mathbf{a}}_{c}$ to achieve the stable point of CE if and only if 
\begin{align}
\label{CE}
\sum\limits_{\mathbf{A}_{-c}\in\mathcal{A}_{-c}}P\left(\mathbf{a}_{c},\mathbf{A}_{-c}\right)\cdot \left[U_{c}(\mathbf{a}_{c},\mathbf{A}_{-c})-U_{c}(\tilde{\mathbf{{a}}}_{c},\mathbf{A}_{-c}) \right]\geq 0,
\end{align}
where $\mathbf{A}_{-c}$ represents the matrix of transmission strategies taken by all CSCs except for CSC $c$. On top of that, the action space formed by all the CSCs except for CSC $c$ is expressed as $\mathcal{A}_{-c}=\prod_{i\neq c}\mathcal{A}_{i}$. 
\end{definition}
It can be observed from (\ref{CE}) that players will coordinate their actions with each other by exchanging the probability distribution of strategies cooperatively. On the contrary, another well-known concept for analyzing the chosen strategies is NE, where each player is inclined to selfishly decide its actions. Note that NE is a point inside CE considering the extreme case that different strategies are independent. Accordingly, a better overall welfare of the players can be intuitively reached on CE compared to the strategies on NE. It induces that higher network EE can be achieved by the CE approach with the cooperation between CSCs. The following theorem prove the existence of CE in our considered small cell networks.
\begin{theorem} \label{thm1}
A CE exists in the proposed resource competition game between the small cell networks.
\end{theorem}
\begin{proof}
Since there is a finite number $C$ of CSCs to choose discrete and finite action space in the proposed CRC game $\mathcal{G}$, it is apparently that $\mathcal{G}$ is a finite game. It has been proved by Theorem 1 in \cite{34} that there exists a CE in every finite game. Therefore, the existence of CE in the proposed CRC game $\mathcal{G}$ can be certified.
\end{proof}
When the subchannels and RSC power are allocated appropriately to UEs, the stable condition, i.e., the Pareto optimum, can be reached by the system. Under the stable point of Pareto optimum, there is no player capable of acquiring higher reward since it potentially causes losses to others.
\begin{theorem} \label{thm2}
With proper strategies of subchannel assignment and transmit power allocation chosen by the CSCs, the Pareto optimum exists in the proposed resource competition game $\mathcal{G}$.
\end{theorem}
\begin{proof}
It has been showed by Theorem 1 that there must exist a CE for the proposed resource competition game $\mathcal{G}$. Moreover, (\ref{CE}) induces that each player can reach its maximum expected utility on the CE. Hence, the sum of total expected utilities of all players can be achieved by choosing the correlated strategies. The existence of Pareto optimum can be proved by contradiction as follows. If there does not exist Pareto optimum in the proposed game $\mathcal{G}$, there must not exist the correlated strategies satisfying (\ref{CE}) and the CE must not exist, which obviously contradicts Theorem \ref{thm1}. Consequently, there must exist the Pareto optimum in proposed game $\mathcal{G}$.
\end{proof}

\subsection{Operational Process for CRC Scheme}
In this subsection, a distributed learning algorithm based on regret matching procedure \cite{35} is adopted for each CSC to determine its correlated strategy in order to obtain the CE in proposed $\mathcal{G}$. Each CSC will select the new strategy of subchannel and power allocation considering the regret value of others not employing strategies in the history. In other words, the higher regret value of non-employed strategies indicates higher probability for the CSC to deviate from its current strategy.

\subsubsection{Initialization}
The allocation of subchannels, transmit power and user association will be initially configured from the TURA scheme within the localized small groups. Each CSC calculates its own utility function by (\ref{U_{c}}) from the subchannels, transmit power of RSC and user association, which are $\mathbf{\Phi_{c}},\mathbf{\Psi_{c}},\mathbf{P_{c}}$ determined in the operational process of TURA scheme.
\subsubsection{Evaluation of Regret Value}
Each CSC determines its potential taken strategies from the regret value. Given a history of adopted strategies, the average regret value of CSC $c$ at time $t$ can be expressed as
\begin{align}
\label{Regret value}
R_{c}^{t}(\mathbf{a}_{c},\hat{\mathbf{a}}_{c})=\text{max}\{{D_{c}^{t}(\mathbf{a}_{c},\hat{\mathbf{a}}_{c}),0}\}.
\end{align}
The $R_{c}^{t}$ given in (\ref{Regret value}) shows that the average regret value depends on the strategy played in the past, $\mathbf{a}_{c}$, and the unused strategy, $\hat{\mathbf{a}}_{c}$, where both $\mathbf{a}_{c}$ and $\hat{\mathbf{a}}_{c} \in \mathcal{A}_{c}$. The term $D_{c}^t$ represents the quantity of difference between the EE for CSC $c$ at time $t$ to adopt the transmission alignments of $\mathbf{a}_{c}$ and $\hat{\mathbf{a}}_{c}$ for subchannel assignment and RSC transmit power allocation. As a result, $D_{c}^{t}$ can be formulated as
\begin{align}
\label{Difference of strategies}
D_{c}^t(\mathbf{a}_{c},\hat{\mathbf{a}}_{c})=\frac{1}{t}\sum\limits_{\substack{\tau \leq t}}[U_{c}(\hat{\mathbf{a}}_{c},\mathbf{A}_{-c}^\tau)-U_{c}^\tau(\mathbf{A}^\tau)].
\end{align}

\subsubsection{Probability Exchanging and Strategy Updating}
The CSCs will choose the strategies with the most profit of utility. After evaluating $R_{c}^{t}$ of potential strategies, the probabilities that the strategies are prone to be selected are updated as follows.
\begin{align}
\label{Probability update}
\mathbf{w}_{c}^t(\hat{\mathbf{a}}_{c}^{t}) = \left\{
\begin{array}{l l}\frac{1}{\xi}R_{c}^{t}(\mathbf{A}_{c}^{\tau},\hat{\mathbf{a}}_{c}^{t}),
&\hat{\mathbf{a}}_{c}^{t}\neq \mathbf{A}_{c}^{\tau},\\
1-\sum\limits_{\substack{\hat{\mathbf{a}}_{c}^{t}\in \mathcal{A}_{c} \\
\hat{\mathbf{a}}_{c}^{t}\neq \mathcal{A}_{c}^{\tau}}}\frac{1}{\xi}R_{c}^{t}(\mathbf{A}_{c}^{\tau},\hat{\mathbf{a}}_{c}^{t}), &  \hat{\mathbf{a}}_{c}^{t} = \mathbf{A}_{c}^{\tau}. \
\end{array}
  \right.
\end{align}
The term $\xi$ in (\ref{Probability update}) is a non-negative and large enough number to normalize the summation of probabilities for different strategies to one. Each CSC will exchange their probabilities of potential strategies with each other. However, the error may occur due to asynchronous information exchanging from other CSCs. Therefore, the exchanged probability vector as shown in Fig. \ref{Fig:OP} can be represented as
\begin{align} \label{exchpro}
\mathbf{\tilde{w}}_{-c}^{t}(\hat{\mathbf{a}}_{-c}^{t})=\mathbf{w}_{-c}^{t}(\hat{\mathbf{a}}_{-c}^{t}) \cdot (1+\rho\Delta h),
\end{align}
where $\mathbf{w}_{-c}^{t}$ is the set for probability of potential strategies from all the CSCs except for CSC $c$. Moreover, the parameter $\rho \in [0,1]$ is the error ratio due to asynchronous exchanged information, and $\Delta h \sim \mathcal{CN}(0,1)$. As a result, each CSC can determine its strategy for transmission alignment as
\begin{align}
\label{action update}
\mathbf{a}_{c}^{t+1}=\arg\underset{\hat{\mathbf{a}}_{c}^{t}\in\mathcal{A}_{c}}{\max}\mathbf{w}_{c}^{t}(\hat{\mathbf{a}}_{c}^t).
\end{align}

\subsubsection{Convergence}
The ultimately determined strategy of each CSC will be achieved according to the following criterion as
\begin{align}
\label{convergence of game}
\frac{|U_{c}^{t+1}-U_{c}^{t}|}{U_{c}^{t}}\leq \theta^{\text{th}},\;\;\quad\forall c.
\end{align}
The subchannel assignment and transmit power of RSCs will be allocated by their serving CSC within the localized C-RAN when the criteria of (\ref{convergence of game}) meets. Otherwise, the CSCs will continue performing the algorithm until the condition of (\ref{convergence of game}) is satisfied. The detailed procedure for adapted algorithm is described in Algorithm \ref{alg:CRC}.
\begin{algorithm}[!tb]
\caption{Proposed CRC Algorithm}
\SetAlgoLined
\DontPrintSemicolon
\label{alg:CRC}
\begin{algorithmic}[1]
\FOR{$s=1,\dots,S$}

\STATE Set the initial iteration as $t=1$
\STATE CSC $c$ initializes its strategy $\mathbf{a}_{s}^1$ of subchannel assignment and transmit power allocation for the users according to the $\mathbf{\Phi_{c}},\mathbf{\Psi_{c}},\mathbf{P_{c}}$ in TURA scheme
\REPEAT
	\STATE Evaluate average regret value of different strategies based on its own utility function of EE by calculating (\ref{Regret value})
	\STATE Update the probability for the strategy according to (\ref{Probability update})
	\STATE Exchange the probability of potential strategies with other CSCs according to $\eqref{exchpro}$
	\STATE Choose the strategy for iteration $t+1$ given the probability distribution of $w_{c}^{t}(\mathbf{a}_{c}^{t})$ by $\mathbf{a}_{c}^{t+1}=\arg\underset{\hat{\mathbf{a}}_{c}^{t}\in\mathcal{A}_{c}}{\max}\mathbf{w}_{c}^{t}(\hat{\mathbf{a}}_{c})$
	\STATE $t=t+1$
\UNTIL{Convergence: $\frac{|U_{c}^{t+1}-U_{c}^{t}|}{U_{c}^{t}}\leq \theta^{\text{th}}$}
	\ENDFOR
\end{algorithmic}
\end{algorithm}

\section{Performance Evaluations} \label{sec:PE}

\begin{table}
\begin{center}
\footnotesize
\caption {Main System Parameters}
    \begin{tabular}{ll}
        \hline
        Parameters of system& Value \\ \hline \hline
        
        Number of subchannels $N$ & 50 \\
        Subchannel bandwidth $W$ & 360 kHz \\
        Carrier frequency & 2.6 GHz \\
        Noise power $N_0$ & $-174$ dBm/Hz \\
        Circuit power $P^{\rm (CA)}$/$P^{\rm (CS)}$ & 6.8/4.3 W \\
        Signal power overhead $P^{\rm (O)}$ & 1 dBm \\
        Maximum transmit power $P^{\rm (max)}$ & 20 dBm \\
        Minimum data rate requirement $R^{\rm (min)}$ & 10 Mbps \\ 
        Coverage area & 90 $\times$ 90 (meter$^{2}$)\\ \hline
        Parameters of HARM scheme & Value \\ \hline \hline
        Maximum number of iterations $T$ & 600\\
        Number of candidate solutions $I$ & 40\\
        Penalty factor $\alpha$ & 1.5 \\
        Searching range $\beta^{\rm (max)}$/$\beta^{\rm (min)}$ & 1.2/0.5\\
        Convergence threshold $F^{\rm (th)}$/$\theta^{\rm {th}}$ & $10^{-4}$/$10^{-4}$\\
        \hline
    \end{tabular} \label{Main_Parameter}
\end{center}
\end{table}

In this section, the performance evaluation of proposed HARM, TURA, and CRC will be provided and illustrated through simulation. Consider a small cell network, the RSCs are deployed in grid in a square coverage and users are uniformly distributed within the coverage area. Each user is located at a minimum distance of 2 meters from each RSC. The network channel model and system parameters are considered based on \cite{36} and \cite{37}, where the case of indoor dense urban information society in \cite{37} is taken as the reference. The channel fading is composed of pathloss and Rayleigh fading. The pathloss model between RSCs and users depends on the sight conditions such as line of sight (LoS) and non-line of sight (NLoS). It can be represented as $PL=18.7\log_{10}(d_{s,k})+46.8+174.32$ in LoS and $PL=36.8\log_{10}(d_{s,k})+43.8+174.32+5(n_{w}-1)$ in NLoS, where $d_{s,k}$ is the distance between RSC $s$ and user $k$ in meters and $n_{w}$ is the number of walls between transmitter and receiver. The Rayleigh fading is modeled as independent and identically distributed Gaussian distribution random variable. The main parameter setting is listed in Table \ref{Main_Parameter}. Moreover, if not mentioned specifically, the data rate of each user, maximum transmit power of RSC and fronthaul capacity of RSC are respectively set to be identical as $R_{c,k}^{\rm (min)}=R^{\rm (min)}$, $P_{c,s}^{\rm (max)}=P^{\rm (max)}$ and $B_{c,s}^{\rm{max}}=B^{\rm{max}}, \forall c,k,s$. All simulation results are averaged over random user locations and channel conditions according to Monte Carlo runs.

\begin{figure}
\centering
\includegraphics[width=3.3in]{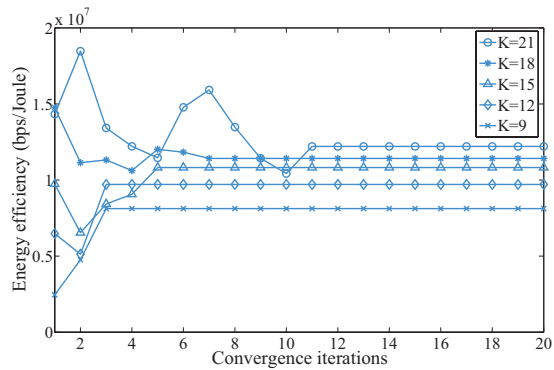}
\caption{Energy efficiency versus number of iterations in proposed CRC scheme under $S=6$ and $B^{\rm (max)}=20$ Mbps.} \label{Fig.S4}
\end{figure}

\subsection{Convergence of CRC Scheme}
\label{Convergence_CRC}
In this subsection, the convergence of proposed CRC scheme is studied in Fig. \ref{Fig.S4} under different number of traffic loads $K$. Note that the CRC scheme mentioned in Section \ref{sec:CRC} is designed to perform resource competition and interference management among multiple C-RANs. We consider that each C-RAN consist of a small cell with synchronous ideal backhaul which integrates a CSC and an RSC for the comprehensive capability to perform network functions. Accordingly, the proposed CRC scheme will be conducted for resource competition at small cells. It can be seen form Fig. \ref{Fig.S4} that the CRC algorithm can quickly converge to its CE and EE under few iterations by adopting the regret-based learning algorithm. More iterations are required to achieve the CE and converged EE in the proposed CRC scheme when the number of users is increased due to the reason that there exist more available strategies for resource allocation to be chosen.

\subsection{Performance of Proposed TURA Scheme}
\label{Analysis_TC-UARA}
In this subsection, the simulation results are provided to demonstrate tthe effectiveness of proposed TURA scheme on traffic control, and the impacts from RSC on/off mechanism and the capacity limitation of fronthaul. Furthermore, the impact of signalling overhead for traffic control on EE performance is also discussed. Since the effect of traffic control is a critical concern in this paper, a benchmark of RSRP-based user association and resource allocation (RURA) method is considered for performance comparison, where it adopts the identical power and subchannel allocation as TURA along with RSRP-based user association.

\begin{figure}
\centering
\includegraphics[width=4in]{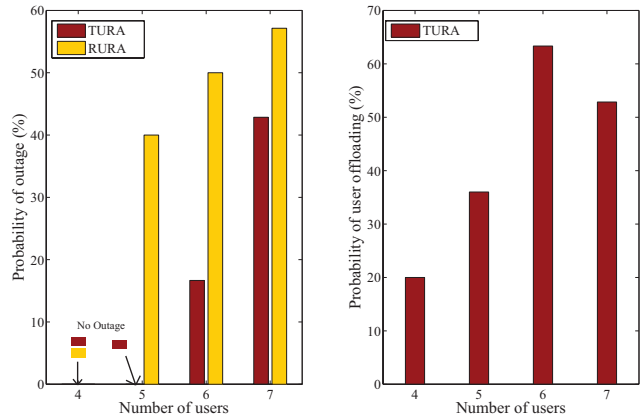}
\caption{Traffic control effect: Probability of outage (left subplot) and user offloading (right subplot) versus different number of users under $S=4$ RSCs and $B^{\rm (max)}=\{10,10,20,10\}$  Mbps.} \label{Fig.S1}
\end{figure}

\begin{figure}
\centering
\includegraphics[width=3.3in]{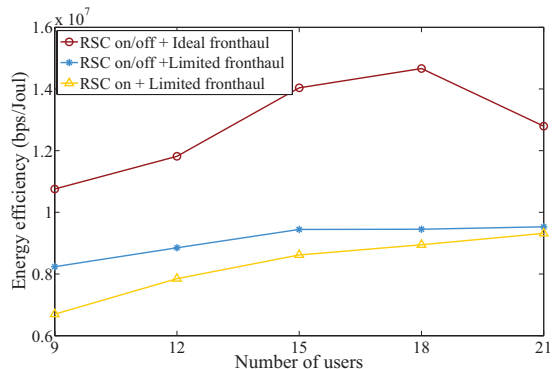}
\caption{Performance comparison for three seanariors in proposed TURA scheme: Energy efficiency versus number of users under $S=6$ and $B^{\rm (max)}=20$ Mbps.} \label{Fig.S2_1}
\end{figure}

\subsubsection{User Offloading and Traffic Control of TURA Scheme}
Fig. \ref{Fig.S1} depicts the relationship between QoS satisfaction and traffic control effect in term of outage probability and traffic control probability respectively. In this figure, a simple case with the number of RSC, $S=4$, is used to observe the effect of traffic control. The capacities of fronthaul are set as $B^{\rm (max)}=\{10,10,20,10\}$ megabit per second (Mbps). A single user is associated with each RSC respectively at the beginning, and then another user is added into the network at a time, i.e, $K=\{4,5,6,7\}$. The left subplot of Fig. \ref{Fig.S1} shows that there is no outage caused in both TURA and RURA schemes with $K=4$ since the capacity of fronthaul is sufficient for the RSC to satisfy the QoS of each user. Furthermore, the outage probability is lower in proposed TURA scheme compared to RURA method with increasing number of users in the network, which illustrates the merits of TURA scheme on user offloading as shown in the right subplot to overcome the restriction of insufficient fronthaul capacity for RSCs. Note that the probability of  traffic control decreases from $N = 6$ to $N = 7$ when the network traffic is overloaded since the total capacity of fronthaul is limited and there is no appropriate configuration of user association to satisfy the data rate requirements of users. Therefore, Fig. \ref{Fig.S1} demonstrates the traffic control capability of improving QoS satisfaction with the limitation of fronthaul capacity. 

\begin{figure}
\centering
\includegraphics[width=3.3in]{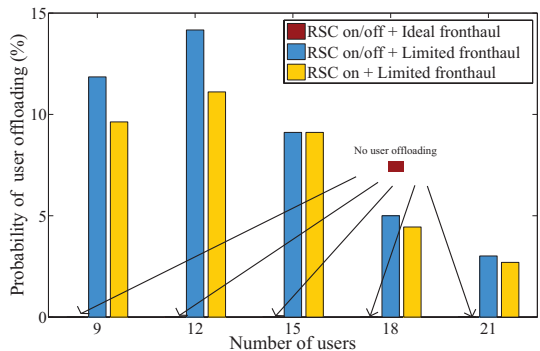}
\caption{Performance comparison for three seanariors in proposed TURA scheme: probability of user offloading versus number of users under $S=6$ and $B^{\rm (max)}=20$ Mbps.} \label{Fig.S2_2}
\end{figure}

\begin{figure}
\centering
\includegraphics[width=3.3in]{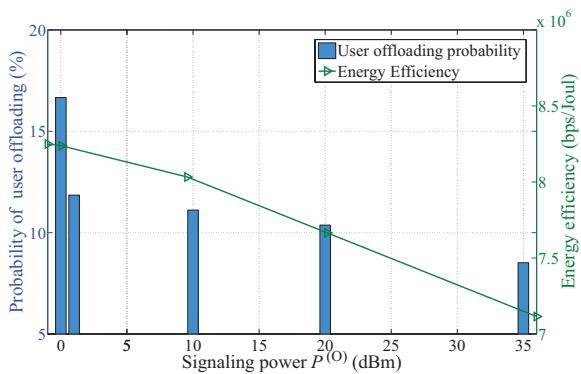}
\caption{Energy efficiency and probability of user offloading versus signaling power $P^{\rm (O)}$ with proposed TURA scheme under $S=6$ and $B^{\rm (max)}=20$ Mbps.} \label{Fig.S3}
\end{figure}

\subsubsection{Impact of RSC on/off Mechanism and Fronthaul Capacity Limit}
In Fig. \ref{Fig.S2_1}, the relationship between EE and different traffic loads $K=\{9,12,15,18,21\}$ under $S=6$ and $B^{\rm (max)}=20$ Mbps is illustrated for three scenarios. The scenario of RSC on/off with unlimited fronthaul capacity provides a better performance on EE over the scenario of RSC on/off  with limited fronthaul since the data rate of user are not restricted by the capacity of fronthaul. Additionally, the performance comparison between the scenarios of RSC on/off with limited fronthaul and RSC on with limited fronthaul depicts the merit of RSC on/off mechanism on the power-saving of RSCs to enhance network EE. Furthermore, it can be observed that the effectiveness of mechanism of RSC on/off mechanism is revealed especially under lower traffic loads, i.e., $K=9$, where the RSCs tend to be turned off so as to conserve energy. Fig. \ref{Fig.S2_2} illustrates the probability of user offloading under different traffic loads. It can be observed that there is no traffic control performed in the scenario of RSC on/off with ideal fronthaul because sufficient capacity is available to support required data rate for the fronthaul. The RSC on/off mechanism provides impact on traffic control since power can be conserved by means of offloading the user to other RSCs and turning off the RSC without serving users. Given that lower amount of required traffic will lead to higher tendency for the RSCs to be turned off by conducting traffic control, the effect of traffic control on energy conservation for RSC power is revealed especially under lower network traffic loads, e.g., $K= \{9,12\}$. It can be concluded that the effect of traffic control is mainly influenced by the limitation of fronthaul capacity. Fig. \ref{Fig.S3} illustrates the probability of user offloading and EE under different signaling power of $P^{\rm (O)}=\{0,1,10,20,35\}$ dBm for proposed TURA scheme with the number of RSC $S=6$ and $B^{\rm (max)} = 20$ Mbps. Note that the signaling power $P^{\rm (O)}$ is considered in the proposed TURA scheme in order to reduce the ping-pong effect in handover process as mentioned in section \ref{sec:TURA}. It can be observed that both the probability of user offloading and EE decrease with the increased signaling power. The reason is that traffic control is performed either for QoS satisfaction when the fronthaul link is overloaded or for energy conservation. Hence, the network EE decreases with the increasing signalling power due to the user offloading for achieving the minimum data rate requirements.

\begin{figure}
\centering
\includegraphics[width=3.3in]{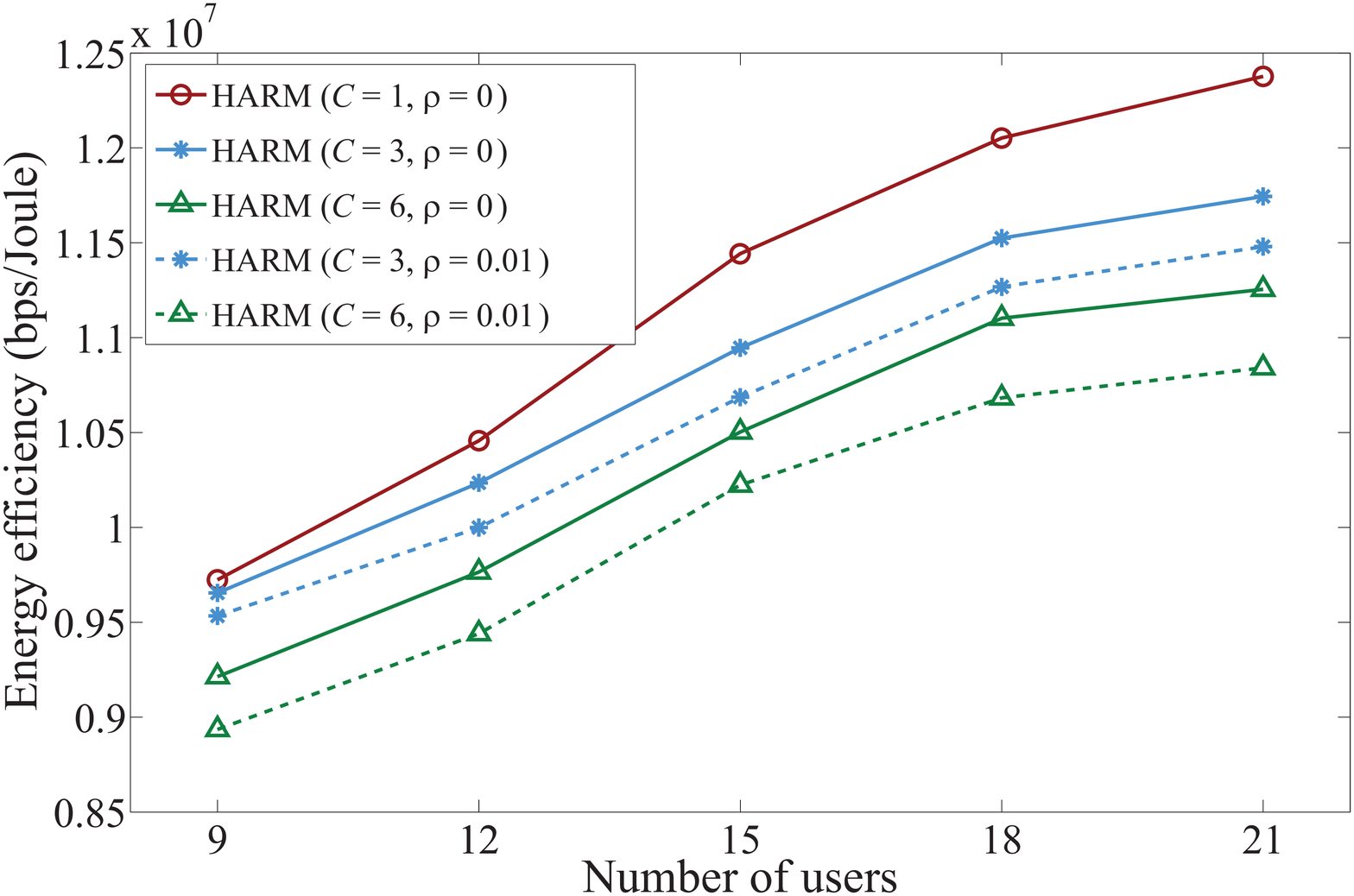}
\caption{Performance of the proposed HARM scheme: Energy efficiency versus number of users under $S=6$, $B^{\rm (max)}=30$ Mbps, $C=\{1,3,6\}$ and $\rho=\{0,0.01\}$.} \label{Fig.S7_1}
\end{figure}

\begin{figure}
\centering
\includegraphics[width=4in]{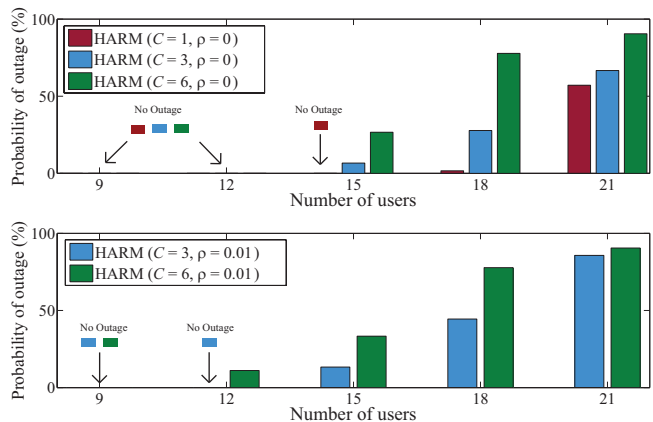}
\caption{Performance comparison of the proposed HARM scheme: Probability of outage versus number of users under $S=6$, $B^{\rm (max)}=30$ Mbps, $C=\{1,3,6\}$ and $\rho=\{0,0.01\}$.} \label{Fig.S7_2}
\end{figure}

\subsection{Performance of Proposed HARM Scheme}
The performance comparison of EE and outage probability of the proposed HARM scheme considering effects of number of localized C-RANs and error ratio of information exchanged versus number of users is shown in Figs. \ref{Fig.S7_1} and \ref{Fig.S7_2}, respectively. Fig. \ref{Fig.S7_1} depicts that the EE of HARM with zero error ratio $\rho=0$ degrades with increasing number of localized C-RANs from $C=1$ to $C=6$. This is because more localized C-RANs divide the network will provoke a comparably more distributed control for resource allocation. Consequently, the coordination among localized C-RANs cannot be performed simultaneously under $C=6$ C-RANs compared to the totally centralized control under $C=1$. With higher error ratio of $\rho=0.01$, the HARM has a lower EE due to asynchronous exchanged information and more power consumption under a large scale network. Furthermore, as illustrated in Fig. \ref{Fig.S7_2}, the outage probability of the proposed HARM scheme with zero error ratio $\rho=0$ asymptotically increases with the increment of localized C-RANs due to non-coordinated interference management and restricted traffic control. Under higher error ratio of asynchronous information of $\rho=0.01$, it potentially brings out higher probability of erroneous resource allocation strategies exchanged between CSCs, which induces inappropriate interference management. In other words, RSCs in the associated localized C-RAN cannot alleviate interference by allocating proper subchannels and transmit power, which also reflects a degradation of EE performance from $\rho=0$ to $\rho=0.01$ as illustrated in Fig. \ref{Fig.S7_1}.

\begin{figure}
\centering
\includegraphics[width=3.3in]{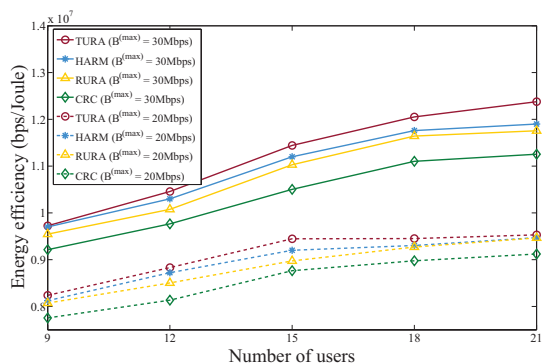}
\caption{Performance comparison between proposed TURA, HARM, RURA, and CRC schemes: Energy efficiency versus number of users under $S=6$ and $B^{\rm (max)}=\{20,30\}$Mbps.} \label{Fig.S5_1}
\end{figure}

\begin{figure}
\centering
\includegraphics[width=4in]{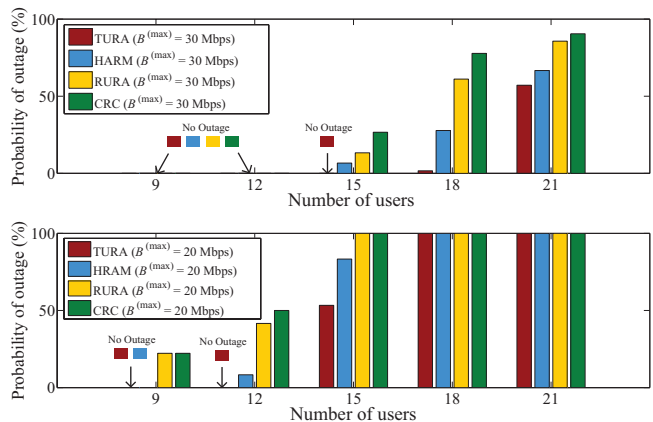}
\caption{Performance comparison between proposed TURA, HARM, RURA, and CRC schemes: Outage probability versus number of users under $S=6$ and $B^{\rm (max)}=\{ 20,30\}$ Mbps.} \label{Fig.S5_2}
\end{figure}

\begin{figure}
\centering
\includegraphics[width=4in]{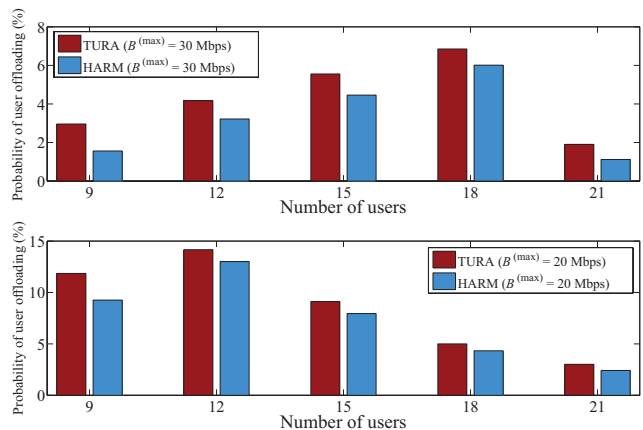}
\caption{Probability of user offloading in proposed TURA and HARM scheme versus number of users under $S=6$ and $B^{\rm (max)}=\{ 20,30\}$ Mbps.} \label{Fig.S5_3}
\end{figure}

\subsection{Performance Comparisons with Existing Methods}
In this subsection, we will evaluate the performance of proposed HARM scheme to observe the integrated effects of both TURA and CRC methods. Note that the centralized control of the proposed TURA scheme considers that all RSCs are managed by a single CSC with upper MAC functions, while the lower MAC functions reside in RSCs with confined fronthaul capacity. Moreover, the performance analyses of HARM are investigated along with the proposed CRC scheme applied in the distributed small cells. Note that the benchmark of RURA adopts RSRP-based user association along with non-adjustable power and subchannel allocation.

\subsubsection{Different Traffic Loads and Fronthaul Capability}
\label{subsubsec:Different traffic}

\begin{figure}
\centering
\includegraphics[width=3.3in]{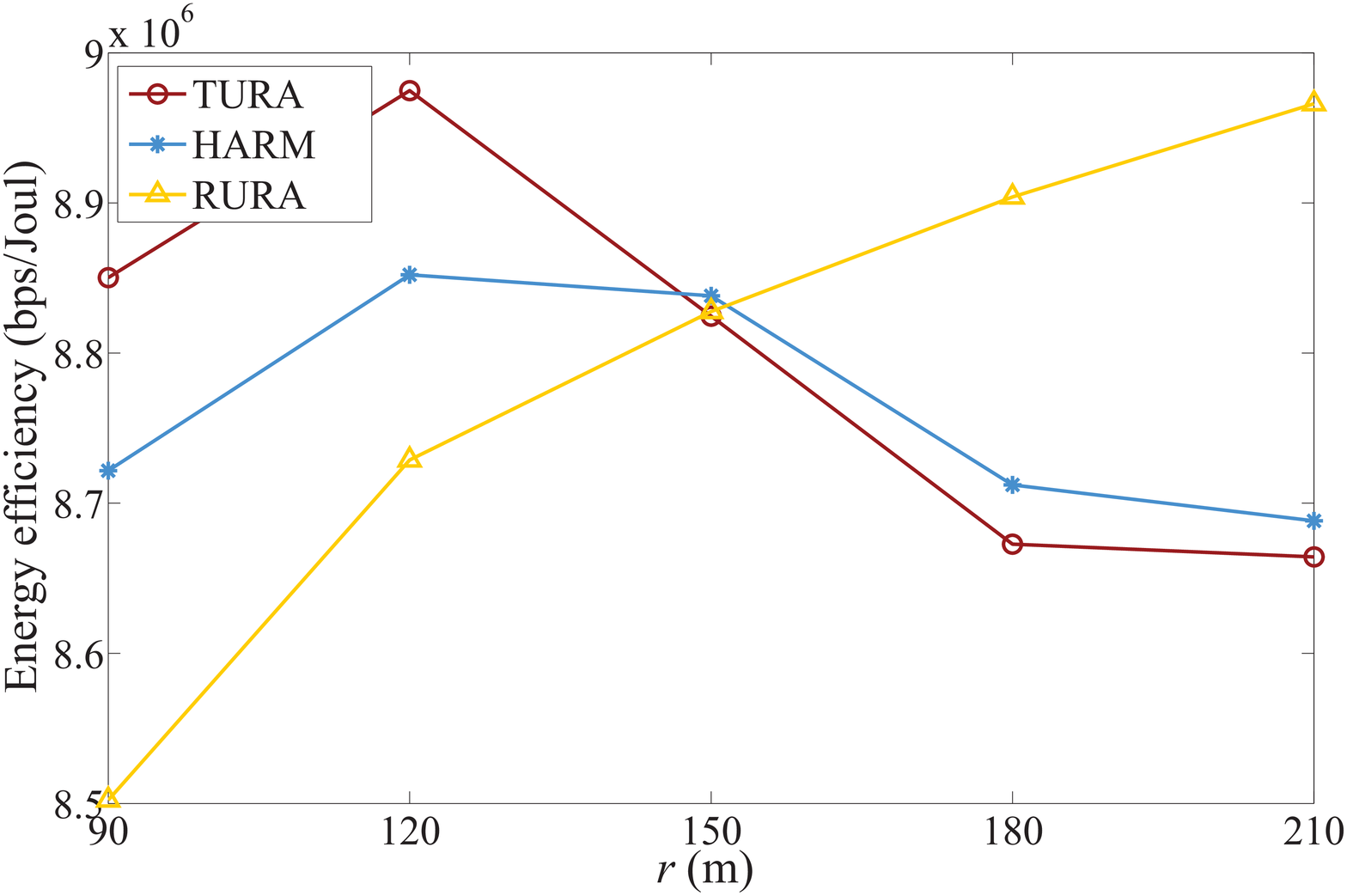}
\caption{Performance comparison between proposed TURA, HARM and RURA schemes: Energy efficiency versus the length of square coverage area $r$ under $K=12$, $S=6$ and $B^{\rm (max)}=20$ Mbps.} \label{Fig.S6_1}
\end{figure}

Fig. \ref{Fig.S5_1} illustrates EE performance of the proposed TURA, HARM, and CRC schemes compared with the benchmark of RURA under different traffic loads $K=\{9,12,15,18,21\}$ and fronthaul capacities $B^{\rm (max)}=\{20,30\}$ Mbps. It can be observed that higher EE can be reached due to sufficient fronthaul capacity of $B^{\rm (max)}=30$ Mbps. The performance of EE in TURA scheme outperform the other schemes in both scenarios due to the centralized control by CSC, which strike a compelling balance of interference management, traffic control under limited fronthaul links. The EE performance of HARM falls a little from TURA since the interference among localized RSCs is not well alleviated as that performed in TURA. Moreover, the users cannot be readily offloaded between different RSCs due to the restricted capacity of fronthaul. Additionally, higher EE is acquired by RURA than that of CRC which is induced by more existing interferences from asynchronous information exchanged and inappropriate resource configuration. The outage probability considering different fronthaul capability is illustrated in Fig. \ref{Fig.S5_2}, whereas Fig. \ref{Fig.S5_3} depicts the probability of user offloading with TURA and HARM schemes under different number of users $K=\{9,12,15,18,21\}$ and fronthaul capacities $B^{\rm (max)}=\{20,30\}$ Mbps. Note that there conducts no traffic control mechanism in the compared schemes of RURA and CRC. It can be observe that the proposed TURA scheme outperforms the other benchmarks due to centralized configuration of resource management and traffic control. On the other hand, HARM has higher outage probability and lower tendency of traffic control because the user cannot be offloaded between two neighbouring RSCs. We can also infer that the outage probability can be substantially reduced with higher fronthaul capacity, i.e., lower outage is achieved under $B^{\rm (max)}=30$ Mbps compared to that under $B^{\rm (max)}=20$ Mbps. However, users suffer from full outage among all schemes due to high traffic loads of $K = \{18, 21\}$ and insufficient fronthaul of $B^{\rm (max)}=20$ Mbps. Moreover, as shown in Fig. \ref{Fig.S5_3}, it depicts that the increased offloading probability is revealed due to asymptotically saturated fronthaul. Furthermore, the probability of user offloading starts to decrease from $K=18$ under $B^{\rm (max)}=30$ and from $K=12$ under $B^{\rm (max)}=20$ since the fronthaul is overloaded which is incapable of supporting the data rate requirements. In addition to QoS satisfaction, user offloading will also be performed to save energy resulting in potential sleep-mode RSCs.

\subsubsection{Impact on Network Density}
\label{subsubsection:ISD}

\begin{figure}
\centering
\includegraphics[width=3.3in]{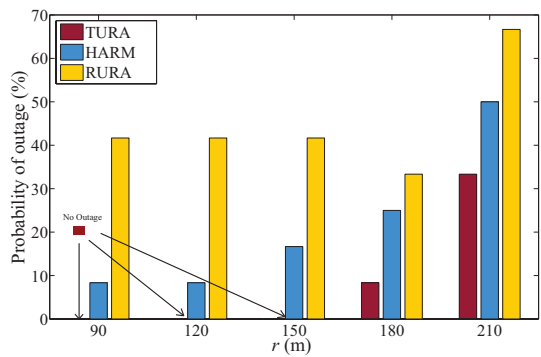}
\caption{Performance comparison between proposed TURA, HARM and RURA schemes: Outage probability versus the length of square coverage area $r$ under $K=12$, $S=6$ and $B^{\rm (max)}=20$ Mbps.} \label{Fig.S6_2}
\end{figure}

\begin{figure}
\centering
\includegraphics[width=3.3in]{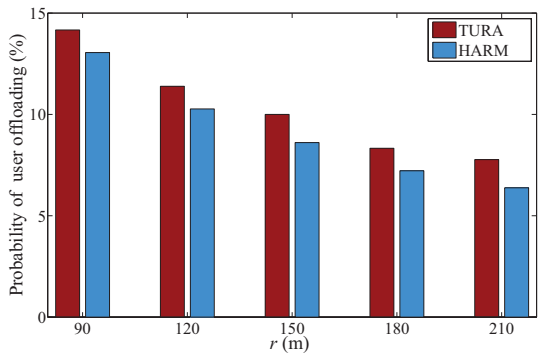}
\caption{Performance comparison between proposed TURA, HARM and RURA schemes: Probability of user offloading versus the length of square coverage area $r$ under $K=12$, $S=6$ and $B^{\rm (max)}=20$ Mbps.} \label{Fig.S6_3}
\end{figure}

We consider that RSCs are deployed within an area of $150 \times 150$ square meters. The performance of the proposed TURA and HARM schemes compared with the benchmark of RURA is demonstrated in Figs. \ref{Fig.S6_1}, \ref{Fig.S6_2} and \ref{Fig.S6_3}. Note that smaller $r$ is regarded as the dense scenario, and vice versa. It can be observed from Fig. \ref{Fig.S6_1} that EE performance of the proposed TURA scheme outperforms both HARM and RURA algorithms under a denser network, i.e., $r\leq 150$, due to significant interference mitigation. On the other hand, inter-RSC interference is intrinsically smaller with longer distance among RSCs, which results in improved EE for RURA method. However, overloaded RSCs will perform traffic control mechanism under the employment of proposed TURA and HARM schemes in order to satisfy the minimum data rate requirements of users. Therefore, as depicted in Fig. \ref{Fig.S6_2}, both TURA and HARM achieve lower outage probability than RURA due to traffic control mechanism, which TURA and HARM respectively reach around $40\%$ and $20\%$ decrement of outage. Nevertheless, TURA has much lower outage than HARM since user-offloading between different C-RANs can be conducted under centralized operation of TURA, which achieves zero outage under the sparser scenario of $90\leq r \leq 150$. This can also be reflected from Fig. \ref{Fig.S6_3} that the user offloading for HARM is less occurred than TURA scheme. Moreover, with shorter distance among RSCs, more frequent offloading mechanism will happen due to more induced interferences.

\section{Conclusions}
In this paper, we conceive hybrid controlled user association and resource management for resolving the problem of subchannel and transmit power allocation for EE maximization considering limited fronthaul capacity under a large scale green C-RANs. Within a localized C-RAN, the CSC centrally performs the proposed TURA scheme for the RSCs to mitigate the intra-group interference and tackle the issue of limited fronthaul capacity for user QoS requirements. Furthermore, the the proposed CRC scheme can analytically achieve the CE and alleviate the inter-group interference among localized RSCs. The Pareto optimum is theoretically proved based on game theory. Moreover, the simulation results have investigated the effect of traffic control by verifying the convergence of CRC scheme. Additionally, the EE performance of the proposed TURA scheme outperforms the other existing methods due to fully-centralized management under the consideration of ideally sufficient fronthaul. Despite the slightly lower EE performance of the proposed HARM scheme than the ideal case of TURA, it is capable of sustaining an appropriate EE than existing schemes under feasible implementation of a practical green RANs.
	

\bibliographystyle{IEEEtran}
\bibliography{IEEEabrv}

\end{document}